\documentclass[11pt,a4paper]{article}
\usepackage[T1]{fontenc}
\usepackage{ae,aecompl}
\usepackage{xcolor}

\usepackage[pdftex,colorlinks=true,citecolor=blue,linkcolor=black,urlcolor=gray,pdfstartview=FitH]{hyperref}
\usepackage{url, fullpage}

\usepackage{amsfonts, amsmath, amssymb}
\usepackage{color}

\newtheorem{theorem}{Theorem}[section]
\newtheorem{corollary}[theorem]{Corollary}
\newtheorem{lemma}[theorem]{Lemma}
\newtheorem{observation}[theorem]{Observation}

\newtheorem{definition}[theorem]{Definition}

\newtheorem{conjecture}[theorem]{Conjecture}

\newcommand{\qed}{\rule{7pt}{7pt}}

\newenvironment{proof}{\noindent{\bf Proof:}\hspace*{1em}}{\qed\bigskip}
\newenvironment{proof-sketch}{\noindent{\bf Proof Sketch:}\hspace*
{1em}}{\qed\bigskip}
\newenvironment{proof-idea}{\noindent{\bf Proof Idea:}\hspace*{1em}}
{\qed\bigskip}
\newenvironment{proof-of-lemma}[1]{\noindent{\bf Proof of Lemma #1}
  \hspace*{1em}}{\qed\bigskip}
\newenvironment{proof-attempt}{\noindent{\bf Proof Attempt}\hspace*
{1em}}{\qed\bigskip}
\newenvironment{proofof}[1]{\noindent{\bf Proof
of #1:}}{\qed\bigskip}

\newcommand{\calp}{\mathcal{P}}

\newcommand{\calb}{\mathcal{B}}
\newcommand{\cala}{\mathcal{A}}
\newcommand{\eps}{\epsilon}

\DeclareMathOperator*{\E}{\mathbb{E}}

\newcommand{\Z}{{\mathbb Z}}
\newcommand{\F}{{\mathbb F}}

\newcommand{\N}{{\mathbb N}}

\newcommand{\R}{{\mathbb R}}
\newcommand{\C}{{\mathbb C}}

\newcommand{\eqdef}{{\stackrel{\rm def}{=}}}

\newcommand{\zo}{{\{0,1\}}}

\newcommand{\ignore}[1]{}
\newcommand{\email}[1]{\textcolor{gray}{{\tt #1}}}

\newcommand{\indd}{\mathrm{indd}}
\newcommand{\indm}{\mathrm{indm}}
\newcommand{\cali}{\mathcal{I}}
\newcommand{\calg}{\mathcal{G}}
\newcommand{\expo}[1]{\mathsf{e}\left(#1\right)}

\begin{document}

\title{Testing Low Complexity Affine-Invariant Properties}
\author{Arnab Bhattacharyya\thanks{Center for Computational
    Intractability. Supported by NSF Grants CCF-0832797, 0830673, and 0528414.}\\Princeton University\\
    \email{arbhat@gmail.com}
\and Eldar Fischer\thanks{Faculty of Computer Science. The research leading to these results has received funding from the ERC 7th Framework Programme grant number 202405.}\\Israel Institute of Technology\\ \email{eldar@cs.technion.ac.il}
  \and Shachar Lovett\thanks{School of Mathematics. Supported by NSF grant DMS-0835373.}\\Institute for Advanced Study\\ \email{slovett@math.ias.edu}}

\maketitle

\begin{abstract}
Invariance with respect to linear or affine transformations
of the domain is arguably the most common symmetry exhibited by natural
algebraic properties.  In this work, we show that any {\em low complexity}
affine-invariant property of multivariate functions over finite fields is
testable with a constant number of  queries. This immediately reproves,
for instance, that the Reed-Muller code over $\F_p$ of degree $d < p$ is
testable, with an argument that uses no detailed algebraic information
about polynomials, except that low degree is preserved by composition
with affine maps.

The complexity of an affine-invariant property $\calp$ refers to the maximum
complexity, as defined by Green and Tao (Ann.~Math.~2008), of the sets
of linear forms used to characterize $\calp$. A more precise statement
of our main result is that  for any fixed prime $p \geq 2$ and fixed
integer $R\geq 2$, any affine-invariant property $\calp$ of functions
$f: \F_p^n \to [R]$ is testable, assuming that the complexity of the
property  is less than $p$. Our proof involves developing analogs of
graph-theoretic techniques in an algebraic setting, using tools from
higher-order Fourier analysis.
\end{abstract}
\thispagestyle{empty}
\setcounter{page}{0}
\thispagestyle{empty}
\newpage

\section{Introduction}\label{sec:intro}
The field of property testing, as initiated by \cite{BLR,BFL} and
defined formally by \cite{RS,GGR}, is the study of algorithms that
query their input a very small number of times and with high
probability decide correctly whether their input satisfies a given property or
is ``far'' from satisfying that property.  A property is called {\em
testable}, or sometimes {\em strongly testable} or {\em locally
testable}, if the number of queries can be made independent of the
size of the object without affecting the correctness probability.
Perhaps surprisingly, it has been found that a large number of
natural properties satisfy this strong requirement; see e.g.\
the surveys \cite{FischerSurvey, RubinfeldICM,
RonSurvey09, SudanSurvey} for a general overview.

A fundamental problem in the area is then to find a combinatorial
{\em characterization} of the testable properties. The
characterization problem was explicitly raised even in the early work
of \cite{GGR}, and for dense graphs it was addressed in a long
series of works culminating in \cite{AFNS06} and \cite{BCLSSV06}.

In this work, we make steps towards such a
characterization for the class of affine-invariant properties of
multivariate functions over finite fields. Before stating our results,
let us define some useful notions that will be helpful to know
throughout this paper.

\subsection{Testability and  Invariances}
Fix a prime $p \geq 2$ and an integer $R \geq 2$ throughout.
Given a property $\calp$ of functions in $\{\F_p^n \to [R]\}$,
we say that $f : \F_p^n \to [R]$ is {\em $\eps$-far} from $\calp$ if
$\min_{g \in \calp} \Pr_{x \in \F_p^n}[f(x) \neq g(x)] > \eps$,
and we say that it is {\em $\eps$-close} otherwise.
$\calp$ is said to be {\em testable} (with one-sided error)
if there is a function $q: (0,1) \to \Z^+$ and an algorithm $T$ that,
given as input a parameter $\eps \in (0,1)$ and oracle access to a
function $f: \F_p^n \to [R]$, makes at most $q(\eps)$ queries to the
oracle for $f$, always accepts if $f \in \calp$ and rejects
with probability at least $2/3$ if $f$ is $\eps$-far from $\calp$.

As an example of a testable property, let us recall the famous
result by Blum, Luby and Rubinfeld \cite{BLR} which started off
this whole line of research. They showed that for testing whether a
function $f: \F_p^n \to \F_p$ is linear or whether it is $\eps$-far
from linear, it is enough to  query the value of $f$ at only
$O(1/\eps)$ points of the domain.

Linearity, in addition to being testable, is also an example of a
{\em linear-invariant} property. We say that a property $\calp
\subseteq \{\F_p^n \to [R]\}$ is linear-invariant if it is the case
that for any $f \in \calp$ and for any linear transformation $L:
\F_p^n \to \F_p^n$, it holds that $f\circ L \in \calp$.  Similarly, an
{\em affine-invariant} property is closed under composition with
affine transformations $A: \F_p^n \to \F_p^n$ (an affine transformation
$A$ is of the form $L+c$ where $L$ is linear and $c$ is a constant).
The property of a function $f: \F_p^n \to \F_p$ being affine is testable
by a simple reduction to \cite{BLR}, and is itself affine-invariant. Other well-studied
examples of affine-invariant (and hence, linear-invariant) properties
include Reed-Muller codes (in other words, bounded degree
polynomials) \cite{BFL, BFLS,FGLSS,RS, AKKLR}, homogeneous polynomials of
bounded degree \cite{KS08}, and subspace juntas \cite{VX11}.

In general, invariance under a large group of symmetries seems to be a
common trait of mathematically natural properties, and in particular,
affine invariance underlies most interesting properties that one would
classify as ``algebraic''. Kaufman and Sudan in \cite{KS08} made
explicit note of this phenomenon and urged a study of the testability of
properties with focus on their invariance. In their paper, Kaufman and
Sudan showed that {\em linear} affine-invariant properties are
automatically testable but left open the general question. Note that
arbitrary affine-invariant properties are not testable; in fact,
testing a random affine-invariant property requires querying nearly
all of the domain. So, the question becomes: what is the minimal set
of restrictions an affine-invariant property must satisfy in order to
be testable? In order to state the conjectured answer to this
question, as well as our progress here, we need to introduce some
more notions.

\subsection{Hereditariness and Induced Affine Constraints}
We now introduce the subclass of affine-invariant properties which, we
believe, captures every property testable with a $1$-sided error test.
\begin{definition}[Affine subspace hereditary properties]
An affine-invariant property $\calp$ is said to be {\em affine subspace
  hereditary} if for any  $f: \F_p^n \to [R]$ satisfying $
\calp$, the restriction of $f$ to any affine subspace of $\F_p^n$ also
satisfies $\calp$.
\end{definition}
Affine subspace hereditariness thus provides something like a
uniformity condition, relating the definition of the property for
different values of $n$. Specializing the conjecture in \cite{BGS10}
for linear-invariant properties to affine-invariant properties gives
the following:
\begin{conjecture}[\cite{BGS10}]\label{conj:main}
Any affine subspace hereditary property is testable with $1$-sided error.
\end{conjecture}
Moreover, \cite{BGS10} show that {\em every} affine-invariant property
testable by a ``natural'' tester is very ``close'' to an affine
subspace hereditary property\footnote{We omit the technical definitions
of ``natural'' and ``close'' here, since they are unimportant
here. Informally, the behavior of a ``natural'' tester is independent of the
size of the domain and ``close'' means that the property deviates from
an actual affine subspace hereditary property on functions over a
finite domain. See \cite{BGS10} for details, or \cite{AS08} for
the analogous definitions in a graph-theoretic
context.}. In fact, resolving Conjecture \ref{conj:main} would yield a
combinatorial {\em characterization} of the (natural) one-sided testable
affine-invariant properties, similar to the characterization for dense
graph properties \cite{AS08}.

Before proceeding, let us give some examples of affine subspace
hereditary properties in order to build intuition about how to test them.  Consider the property
of being affine, by which we mean here that the function is a polynomial
of degree at most $1$. This is clearly an affine-invariant hereditary
property. As we remarked earlier, the property is known to be
testable. Note that here, we could also have defined being affine as the
condition of satisfying the identity $f(x)-f(x+y)-f(x+z)+f(x+y+z) = 0$
for every $x,y,z\in \F_p^n$. This is a ``local'' characterization of
being affine, in the sense that the functional equation does not depend on
the value of $n$. Moreover, this characterization
automatically suggests a $4$-query test: pick random $x, y, z \in
\F_p^n$ and check whether the identity holds or not for that choice of
$x, y, z$.

More generally, consider the property of being a
polynomial of degree at most $d$, for some fixed positive integer $d$.
Again, the property is clearly affine subspace hereditary. It is also
known to be testable \cite{AKKLR} over finite fields. And just as in
the case of linearity, the test arises out of a local
characterization for degree $d$: the $(d+1)$th derivative in every
$d+1$ directions at every point should be $0$. The test is then to choose a
random point and random $d+1$ directions and to check whether the $(d+1)$th
derivative in the chosen directions at the chosen point is $0$ or not.

In fact, one can describe any affine subspace hereditary property
using (finitely or infinitely many) such local characterizations. To state this formally, let us
put forth a useful definition.
\begin{definition}[Affine constraints]
\
\begin{itemize}
\item
An {\em affine constraint of size $m$ on $\ell$ variables} is a tuple  $A
= (a_1,\dots,a_m)$ of $m$ linear forms $a_1,\dots, a_m$ over $\F_p$ on
$\ell$ variables, where $a_1(X_1,\dots,X_\ell) = X_1$ and for every $i
\geq 2$, $a_i(X_1,\dots, X_\ell) = X_1 + \sum_{j=2}^\ell c_{i,j}X_j$ where
 each $c_{i,j} \in \F_p$.
\item
An {\em induced affine constraint of size $m$ on $\ell$ variables} is
a pair $(A,\sigma)$ where $A$ is an affine constraint of size $m$ on
$\ell$ variables and $\sigma \in [R]^m$.
\item
Given such an induced affine constraint $(A,\sigma)$, a function $f:
\F_p^n \to [R]$ is said to be {\em $(A,\sigma)$-free} if there exist no
$x_1, \dots, x_\ell \in \F_p^n$ such that $(f(a_1(x_1, \dots, x_\ell)),
\dots, f(a_m(x_1,\dots,x_\ell))) = \sigma$. On the other hand, if such
$x_1, \dots, x_\ell$ exist, we say that {\em $f$ induces $(A,\sigma)$ at
  $x_1, \dots, x_\ell$}.
\item
Given a (possibly infinite) collection $\cala = \{(A^1,\sigma^1),
(A^2, \sigma^2), \dots, (A^i,\sigma^i),\dots\}$ of induced affine constraints, a function $f: \F_p^n
\to [R]$ is said to be {\em $\cala$-free} if it is
$(A^i,\sigma^i)$-free for every $i \geq 1$.
\end{itemize}
\end{definition}

The connection between affine subspace hereditariness and affine
constraints is given by the following simple observation.
\begin{observation}
An affine-invariant property $\calp$ is affine subspace hereditary if
and only if it is equivalent to the property of $\cala$-freeness for
some fixed collection $\cala$ of induced affine constraints.
\end{observation}
\begin{proof}
Given an affine invariant property $\calp$, a simple (though inefficient) way to obtain the set $\cala$ is to let it be the following: For every $n$ and a function $f:\F_p^n$ that is not in $\calp$, we include in $\cala$ the constraint $(A_f,\sigma_f)$, where $A_f$ is indexed by members of $\F_p^n$ and contains $\{a_z(X_1,\ldots,X_{n+1})=X_1+\sum_{i=1}^nz_iX_{i+1}:z=(z_1,\ldots,z_n)\in\F_p^n\}$, and $\sigma_f$ is just set to $f$. From here it is easy to see that the property defined by $\cala$ is contained in $\calp$, while containment in the other direction follows from $\calp$ being affine-invariant and hereditary.

The other direction of the observation is trivial.
\end{proof}

Thus, resolving Conjecture \ref{conj:main} boils down to showing
testability for all $\cala$-freeness properties.

\subsection{Main Result}
We show that $\cala$-freeness is testable as long as all affine
constraints in $\cala$ are of {\em complexity} less than $p$. We next
define the complexity of an affine constraint, and more generally, of
an arbitrary set of linear forms.

\begin{definition}[Cauchy-Schwarz complexity, \cite{GT06}]\label{def:cplx}
Let $\mathcal{L} = \{L_1,\dots,L_m\}$ be a set of linear forms. The
{\em (Cauchy-Schwarz) complexity of $\mathcal{L}$} is the minimal $s$
such that the following holds. For every $i \in [m]$, we can partition
$\{L_j\}_{j \in [m]\setminus \{i\}}$ into $s+1$ subsets such that
$L_i$ does not belong to the linear span of any subset.
\end{definition}

Given this, one can formulate a conjecture that is a weakened version
of Conjecture \ref{conj:main}:

\begin{conjecture}\label{conj:main2}
A property that is given by a collection of induced affine constraints with a global bound
on their complexity is testable with a $1$-sided error.
\end{conjecture}

The following is our main result, which shows the above when
the complexity bound is strictly smaller than the field size.
\begin{theorem}[Main theorem]\label{thm:main}
For any $\eps \in (0,1)$ and for any (possibly infinite) fixed
collection $\cala = \{(A^1,\sigma^1),$ $(A^2, \sigma^2),$
$\dots,$ $(A^i,\sigma^i), \dots \}$ of induced affine constraints such
that each $A^i$ has complexity less than $p$, there is
a function $q_\cala:(0,1)\to \Z^+$ and  a one-sided tester which
determines whether a function $f : \F_p^n \to [R]$ is $\cala$-free or
$\eps$-far from being $\cala$-free, by making at most $q_\cala(\eps)$
queries to $f$.
\end{theorem}

The function $q_\cala$ has a rather horrible, Ackermann function-like,
dependence on $1/\eps$. Our primary concern in this work though is to
establish testability, and we make no effort in improving the growth
of $q_\cala$. We note though that recent work by Kalyanasundaram and
Shapira \cite{KS11} and by Conlon and Fox \cite{CF11}, building on
previous work by Gowers \cite{Gow97}, suggests that the very rapid growth
of the query complexity function is in fact inherent in the nature of the problem.

Let us lastly note that Theorem \ref{thm:main} is quite nontrivial
even when the collection $\cala$ is finite. Indeed, even if $\cala$
consists only of a single induced affine constraint of complexity
greater than $1$, it was not known
previously how to show testability. We give more details about past
work in Section \ref{sec:past}.

\subsection{Overview of the Proof}

To show Theorem \ref{thm:main}, we will in fact show the following statement.
Note that it uses a yet undefined notion of ``conciseness''; for now it suffices
to know that every $\cala$ is equivalent to a concise one, as we will later prove.

\begin{theorem}\label{thm:main2}
Suppose we are given a possibly infinite collection of labeled affine
constraints $\mathcal{A} = \{(A^1, \sigma^1), (A^2, \sigma^2),$
$\dots,$ $(A^i,\sigma^i),$ $\dots\}$ where $\cala$ is concise, every
$A^i$ is of complexity less than $p$ and consists of $m_i$ linear
forms over $\ell_i$ variables, and $\sigma^i \in [R]^{m_i}$ for every $i$.
Then,  there are functions $\ell_\cala(\cdot)$ and
$\delta_\cala(\cdot)$ such that the following is true for any $\eps
\in (0,1)$.  If a function $f: \F_p^n \to [R]$ with  is $\eps$-far from being $\cala$-free,
then $f$ induces at least $\delta_{\cala}(\eps) \cdot p^{n\ell_i}$ many copies of $(A^i,\sigma^i)$ for
some $i$ such that $\ell_i < \ell_\cala(\eps)$.
\end{theorem}

Theorem \ref{thm:main} immediately follows. Consider the following test: choose
uniformly at random $x_1,\dots,$ $x_{\ell_\cala(\eps)}$ $\in$ $\F_p^n$, let $H$
denote the affine space $\{x_1+\sum_{j=2}^{\ell_\cala(\eps)} c_j x_j : c_j
\in \F_p\}$, and check whether $f$ restricted to $H$ is $\cala$-free
or not. By Theorem \ref{thm:main2}, if $f$ is $\eps$-far from
$\cala$-freeness, then this test rejects with probability at least
$\delta_\cala(\eps)$. Repeating the test $O(1/\delta_\cala(\eps))$
times then guarantees a constant rejection probability. And of course,
if $f$ is $\cala$-free, the test always accepts.

Let us now give an overview of our proof of Theorem \ref{thm:main2}. For simplicity of
exposition, assume for now that $\cala$ consists only of a single
induced affine constraint $(A,\sigma)$ where $A$ is the tuple of
linear forms $(a_1,\dots,a_m)$, each over $\ell$ variables, and
$\sigma \in [R]^m$. For $i \in [R]$, let $f^{(i)}:\F_p^n \to \zo$ be
the indicator function for the set $f^{-1}(\{i\})$. Our goal will then be
to show that, when $f$ is $\eps$-far from $(A,\sigma)$-free, then:
\begin{equation}\label{eqn:avg}
\E_{x_1,\dots,x_\ell}\left[f^{(\sigma_1)}(a_1(x_1,\dots,x_\ell)) \cdot
f^{(\sigma_2)}(a_2(x_1,\dots,x_\ell)) \cdots
f^{(\sigma_m)}(a_m(x_1,\dots,x_\ell))  \right] \geq \delta(\epsilon),
\end{equation}
where crucially, $\delta$ is a positive function that does not depend
on $n$. If we could show this, then we would be done since a valid
test would be to repeat the following procedure $O(1/\delta)$ times:
uniformly pick $x_1,\dots,x_\ell \in \F_p^n$ and immediately reject if
$(f(a_1(x_1,\dots,x_\ell)), \dots, f(a_m(x_1,\dots,x_\ell))) =
\sigma$.

Studying averages of products, as in (\ref{eqn:avg}), has been crucial
to a wide range of problems in additive combinatorics and analytic
number theory. Szemer\'edi's theorem about the density of arithmetic
progressions in subsets of the integers is a classic
example. Szemer\'edi's work \cite{Szem75} arguably initiated such
questions in additive combinatorics, but the major development which
led to a more systematic understanding of these averages was Gowers'
definition of a new notion of uniformity in a Fourier-analytic proof
for Szemer\'edi's theorem \cite{Gow01}. In particular, Gowers introduced
the {\em Gowers norm} $\|\cdot\|_{U^d}$ for a parameter $d \geq 1$,
which allows us to say the following about (\ref{eqn:avg}). If
$\|f_1\|_{U^{d+1}} < \eps$ (for some $d$), $f_2,\dots,f_m$ are arbitrary functions
that are bounded inside $[-1,1]$, and $L_1,\dots,L_m$ are arbitrary
linear forms, then 
$\E_{x_1,\dots,x_\ell \in \F_p^n} \left[ \prod_{i=1}^m f_i(L_i(x_1,\dots,x_\ell))\right]$
is at most $\eps$.

This observation leads to the study of {\em decomposition theorems},
that express an arbitrary function as a linear combination of
functions which have either small Gowers norm or are structured in
some sense. This is an extension of classical Fourier analysis over
$\F_p^n$, where a function is expressed as a linear combination of a
small number of characters with high Fourier mass plus a small error
term. To deal with Gowers norm, the ``characters'' need to be
exponentials of not only linear functions, as in classical Fourier
analysis, but of higher degree
polynomials. Approximate orthogonality among these ``characters'' was
established by Green and Tao in \cite{GT07} and by Kaufman and Lovett
in \cite{KL08}. At this stage, one might expect that results by Hatami and Lovett
\cite{HL11,HL11b} can allow us to use orthogonality to approximate the
expectation of the form in (\ref{eqn:avg}).

Unfortunately, the proof does not follow that easily from \cite{HL11}.  There are two
main reasons for this. The first is that the only information we have about
the original function $f$ is  $\eps$-farness from
$(A,\sigma)$-freeness. Information about correlation, as was assumed
in \cite{HL11}, allows more straightforward application of the
higher-order Fourier analytic tools. We use ideas inspired by previous work
on property testing in the dense model, as in \cite{AFKS} and
\cite{AS08a}, to locate regions of the domain in which we are
guaranteed to find at least one induced occurrence of
$(A,\sigma)$. This leads to a new combinatorially flavored
decomposition theorem (Theorem \ref{thm:subatom2}), which may be of
independent interest. 

The second problem we face is one which also arose in a work by
Green and Tao on decomposition theorems (a.k.a., regularity lemmas)
over the integers \cite{GT10}. Namely, the decomposition theorem we
use decomposes an arbitrary function $f: \F_p^n \to \R$ to a sum of
three functions $f_1, f_2, f_3$. $f_1$ consists of the approximate
``characters'' as mentioned above, $f_2$ has small Gowers norm, and
$f_3$ has low $L^2$-norm. Now, the closeness to
orthogonality for $f_1$ and the smallness of the Gowers norm for $f_2$
decreases as a function of the ``complexity'' of the decomposition,
and are thus, essentially negligible for the purposes of the proof. On
the other hand, the bound on the $L^2$-norm for $f_3$ is only
moderately small and cannot be made to decrease as a function of the complexity
of the decomposition. To get around we essentially use a sequence of
two decompositions, and make the norm of the second one
decrease as a function of the complexity of the first,
where we show that this is enough for our purposes.

\subsection{Previous Work}\label{sec:past}
This work is part of a sequence of works investigating the
relationship between invariance and testability of properties. As
described, Kaufman and Sudan \cite{KS08} initiated the
program. Subsequently, Bhattacharyya, Chen, Sudan and Xie \cite{BCSX09}
investigated {\em monotone} linear-invariant properties of functions
$f:\F_2^n \to \zo$, where a property $\calp$ is monotone if it
satisfies the condition that for any function $g \in \calp$, modifying
$g$ by changing some outputs from $1$ to $0$ does not make it violate
$\calp$. Kr\'al, Serra and Vena \cite{KSV12} and, independently,
Shapira \cite{Sha09} showed testability for any monotone linear-invariant
property characterized by a finite number of linear constraints
(of arbitrary complexity).

Progress has been significantly slower for the non-monotone
properties. Bhattacharyya, Grigorescu, and Shapira proved in
\cite{BGS10} that linear-invariant properties of functions in
$\{\F_2^n \to \zo\}$ are testable if the complexity of the property is
$1$. When restricted to affine-invariant properties, the result of
\cite{BGS10} is a special case of the main result here for $p=2$.
The previous works did not explicitly use higher-order Fourier
analysis; \cite{KSV12} and \cite{Sha09} used variants of the
hypergraph regularity lemma which are similar in spirit to higher-order
Fourier analysis, but are somewhat harder to manipulate due to the lack
of analytic tools.

Higher-order Fourier analysis began with the work of Gowers
\cite{Gow98} and parallel ergodic-theoretic work by Host and Kra
\cite{HK05}. Applications to analytic number theory inspired much more
study by Gowers, Green, Tao, Wolf, and Ziegler among others. A book in
preparation by Tao \cite{Tao11} surveys the current theory of
higher-order Fourier analysis. Our work in this paper relies on
decomposition theorems over finite fields of the type first explicitly
described by Green in \cite{Gre07}. We also heavily use decomposition
results by Hatami and Lovett \cite{HL11}, as described in the
previous section. 

At a high level, the argument to prove our main theorem mirrors ideas
used in a sequence of works \cite{AFKS, AS08a, AS08, graphestim,
  AFNS06, BCLSSV06} to characterize the testable graph properties.  In
particular, the technique of simultaneously decomposing the domain
into a coarse partition and a fine partition with very strong
regularity properties is due to \cite{AFKS}, and the compactness
argument used to handle infinitely many constraints is due to
\cite{AS08a}. However, implementing these graph-theoretic techniques
using higher-order Fourier analysis required several new ideas which, we hope,
can be extended to eventually prove Conjecture \ref{conj:main2}.

\subsection{Further research}
We study affine subspace hereditary properties, and show that if they are defined
by affine constraints of low complexity then they are locally testable. There are several
obvious possible generalizations to this work:
\begin{enumerate}
\item Remove the condition that the field size is larger than the complexity of the affine forms,
thus proving Conjecture \ref{conj:main2}; this requires
non-trivial generalizations of several technical lemmas to small fields, and may require new methods.
\item Handle all linear invariant properties (and not just affine invariant properties).
\end{enumerate}
A third generalization, which might be too strong to hold, is to completely remove the bounded complexity assumption on the linear forms, thus proving Conjecture \ref{conj:main}. In several analogs of this line of research in hypergraph testing, this requirement is analogous to requiring bounded uniformity from the hypergraphs, which is implicitly assumed in all previous works on hypergraph testing. It would be thus also be interesting if the full Conjecture \ref{conj:main} can be {\em disproved}.

\ignore{
\subsection{Organization}

In this extended abstract, we describe our argument at an informal
level, highlighting the new tools we develop and the proof techniques
we employ. The formal proof is deferred to the appendix. 
}


\section{Map of the proof}\label{sec:map}

The rest of this section will be devoted to an informal description of
the building
blocks required to prove Theorem \ref{thm:main}, and by extension
Theorem \ref{thm:main2}. We believe that some of these building
blocks, and especially the ``Super Decomposition'' Theorem
\ref{thm:superdecomp} that we describe below, will be of independent
interest. 

In Section \ref{sec:tools} and Section \ref{sec:decomp}, we develop the main technical tools that we will need for our testability proof. Some of the following lemmas and arguments were proved before: Decomposition lemmas (without rank) were implicit in previous works by Green and Tao and explicit in \cite{HL11}; the existence of a refinement of a given rank was first proved in \cite{GT07} (which is combined here with a decomposition lemma); other prior works are cited along with the proofs below.

Our new contributions there lie in the following:
\begin{itemize}
\item Our final ``Super Decomposition'' Theorem \ref{thm:superdecomp}, and its related ``Subcell Selection'' Corollary \ref{cor:subatom}, are new. Their relation to the original decomposition lemma could be thought of as somewhat akin to the relation of the strong graph regularity lemma in \cite{AFKS} to the original regularity lemma of Szemer\'edi.
\item For the subcell selection corollary to work at all, we need to take careful count of when is a refinement of a partition by polynomials syntactic (i.e.\ there is a containment relationship between the polynomials defining the two partitions) or merely semantic (i.e.\ the polynomials may be different but the partitions they define satisfy a combinatorial refinement relationship). We add the accounting of syntactical vs semantic refinements to all the arguments leading up to our super decomposition theorem.
\item We set the entire analysis in a ``robustness'' framework akin to the one developed for graphs in \cite{graphestim}. This streamlines the argument (essentially allowing us to encapsulate and move away iterative refinement arguments), which could get very unwieldy by the time the super decomposition theorem is reached.
\end{itemize}

In Section \ref{sec:count}, we then develop algebraic and combinatorial constructions, that allow us to {\em use} Corollary \ref{cor:subatom} to provide counting type theorems, and in our case the main ``algebro-combinatorial'' Theorem \ref{thm:main2}. The algebraic part mostly involve procedures that calculate the numbers of affine configuration of a given type that satisfy given polynomial constraints; we also prove, using basic algebra, that we can assume the technical condition of $\cala$ being ``concise'', that is not having more variables than conditions in any of its constraints. The combinatorial part is the ``cleanup'' procedure that we describe below.

We now describe the main components of our proofs in detail.

\subsection{Partition by Polynomial Factors}

We generally deal with a function $f:\F_p^n\to \{0,1\}$ (where a larger fixed size range $[R]$ is handled by considering a sequence of functions rather than one function -- see Subsection \ref{subsec:multi}), and would like to partition its domain $\F_p$ into a small number of regions, so that $f$ has certain ``randomness'' properties in every region (or at least most of them). In the broadest terms, we seek algebraic analogs to Szemer\'edi's regularity lemma and its derivatives that have revolutionized graph theory. Recall that Szemer\'edi's lemma partitions the vertex set of the graph so that most vertex set pairs exhibit random-like properties in the bipartite subgraphs that they induce.

The groundwork providing this started with the works of Green and Tao. In general, a function $f:\F_p^n\to\{0,1\}$ can be decomposed to a sum of three real-valued functions. One that is constant on large regions of the input, one that generally takes small values (in terms of its $l_2$ norm), and one that is ``very random'' (in the sense of the Gowers norm). The relevance of the Gowers norm to our arguments is highlighted in Subsection \ref{subsec:norms}.

In an ideal world, the large regions of the input over which we have a constant function should come from a partition of $\F_p^n$ into affine subspaces, but in fact this cannot be the case. The next best thing is to have a partition based on the values of a fixed length sequence of low degree polynomials over $\F_p^n$. These are called {\em polynomial factors} as per Definition \ref{def:factor}, and the regions of $\F_p^n$ of their respective partitions are called {\em cells}.

However, now we need to re-address the question of independence. Standard linear independence would be insufficient to even guarantee that all regions are of similar sizes, let alone provide other ``randomness'' features. For this we use the notion of polynomial {\em rank}, first developed in \cite{GT07}. Subsection \ref{subsec:poly} provides the details about polynomial factors and their rank.

\subsection{Refinements and the Robustness Framework}

For our purpose it is not enough to prove the existence of certain factors, and we will consider a relationships between pairs of factors, namely the {\em refinement} relationship. There are two kinds of refinements. The ``combinatorial'' {\em semantic} refinement notion means that the partition induced by the second factor consists of subsets of the sets of the first factor, while the ``explicit'' {\em syntactic} refinement notion means that the second factor is in fact defined by a sequence of polynomials extending the sequence that defines the first factor. Definition \ref{def:refine} provides the details.

An important measure of a factor with respect to a function $f:\F_p^n\to\{0,1\}$ is its {\em density index}, as per Definition \ref{def:indd}. This was used in previous decomposition proofs, and is analogous to the index of a graph partition used in the proof of Szemer\'edi's regularity lemma and its variants. In Subsection \ref{subsec:refrob} we introduce and analyze the framework of factor {\em robustness}, where a factor is considered robust if it cannot be refined (with respect to a size bound given as a function of the current size) in a way that significantly increases its index. Robust factors, including ones that refine existing factors, exist by a simple argument, Observation \ref{obs:rob}.

The robustness framework greatly simplifies the arguments used to prove the decomposition theorems in Section \ref{sec:decomp}. Where previously such proofs used an iterative argument, basically repeating a construction of a refining factor as long as the factor does not provide the required properties, in the proofs here we start with a robust factor and then show that it provides the required object.

However, we need a factor to be both robust and of high rank. The high rank requirement (also as a function of the factor size) is in fact also provided through an iterative argument resembling the proof of regularity. In Subsection \ref{subsec:robrank} we integrate arguments similar to those originally made in \cite{GT07} to provide Lemma \ref{lem:rankrob}, the driving engine of our decomposition theorems. This lemma provides factor that is both robust and of high rank. Moreover, if we start from an existing factor that is a syntactic refinement of a base factor that also has high rank, then our new robust factor will additionally be a syntactic refinement of the same base factor. This is crucial to our super decomposition theorem, that requires such a refinement to be provided.

\subsection{Decompositions and Super Decompositions}

Chronologically, decomposition theorems for functions $f:\F_p^n\to\{0,1\}$ have progressed in stages. First a weak decomposition theorem was shown, where a factor is found and $f$ is decomposed into a sum of two functions, $f=f_1+f_2$, where $f_1:\F_n^p\to [0,1]$ is constant over every cell of the factor, and $f_2:\F_p^n\to [-1,1]$ has a bounded Gowers norm. In an ideal world we would like $f_2$ to have a bounded $l_2$ norm, as it denotes an ``error'' of some kind, but this is not possible.

However, for the Gowers norm bound to be of any use, it has to be bounded as a decreasing function of the factor size $C$. The next step was then to find a factor and a decomposition $f=f_1+f_2+f_3$, where $f_3$ is an ``error'' term that is of bounded $l_2$ norm (as we originally intended), and $f_2$ now has a Gowers norm that is smaller than the required function of $C$. The proof ``internally'' uses a sequence of two factors, one refining the other, and a corresponding ``iterated argument of iterated arguments''. However here we can encapsulate it through a robustness requirement. We provide the full details in Subsection \ref{subsec:strong}, which culminates in Theorem \ref{thm:strongdecomp}, providing also a rank requirement. It is similar to theorems proved in  previous works, but here we also maintain a syntactic refinement relationship to a base factor, a feature that will be used later.

This brings us to our new super decomposition Theorem \ref{thm:superdecomp}. Its motivation is that for our purpose, we would also need the $l_2$ norm of the error function $f_3$ to decrease as a function of the factor size. This is required because for our analysis of non-monotone properties, we cannot make do with most of the cells of the factor exhibiting a random-like behavior of $f$ -- we would like {\em all} of them to exhibit it. However, such a demand on $f_3$ is clearly not possible.

The solution is then to provide a sequence of two factors, where the second factor is a syntactic refinement of the first. We then decompose $f$ with respect to the second factor, as a sum of a constant-over-cells function $f_1$, a small Gowers norm function $f_2$, and a function $f_3$ whose $l_2$ norm is not small as a function of the second factor, but at least it is small as a function of the first factor. Additionally, we want $f_1$ to be ``faithful'' also with respect to the first factor: That is, if we had decomposed $f$ according to the first factor rather than the second, then the corresponding ``$f_1$ function'' would still be close in most places to the function we got by decomposing according to the second factor.

In the next step of the proof of our main testability theorem, we will pick one ``subcell'', a cell of the second factor, out of every cell of the first factor. We will want most of these cells to be faithful (with respect to $f_1$) and all of them to exhibit the randomness properties. The syntactic refinement relationship in our super decomposition theorem is what allows us to pick these cells in a ``uniform'' manner, as per our subcell selection Corollary \ref{cor:subatom}.

We believe that Theorem \ref{thm:superdecomp} and its proof methods are of independent interest, as they could open up possibilities for more analogies to the big body of knowledge concerning the applications of Szemer\'edi's lemma and its variants for graphs.

\subsection{Function Cleanup}

To find many induced structures in $f$, we restrict ourselves to the ``good'' subcells chosen by use of Corollary \ref{cor:subatom}. However, to find the correct configuration of subcells exhibiting the induced structures, we refer to a modification of $f$ called a {\em cleanup}. The modified $f$ will be close to the original, and hence will still contain an induced structure. This particular structure might not exist in the original $f$, but the way the cleanup is performed, as per Definition \ref{def:clean}, ensures the existence of the corresponding subcell configuration which ``mimics'' the location of the points of the structure (even that it may not actually contain those points). We then use the configuration of subcells with respect to the original $f$ to find our affine structures.

This argument is in fact somewhat analogous to the argument considering forbidden induced subgraphs that appeared first in \cite{AFKS}. The function closeness lemma is Lemma \ref{lem:cleanup}, while the mimicking subcell argument is found in the proof of Theorem \ref{thm:main2} in Subsection \ref{subsec:proof}.

\subsection{Randomness and consistency}

After we find the subcell configuration corresponding to an affine induced structure, we still need to lower-bound the number of actual copies of the structure that it guarantees for $f$. This requires giving a lower bound for the number of actual small affine sets that reside in this configuration, and within them the number of sets for which $f$ has the corresponding values. The second task is in fact accomplished by the function decomposition that we have. For the first task, we build upon works of Hatami and Lovett \cite{HL11b} and of Gowers and Wolf \cite{gowers-wolf-1,gowers-wolf-2} in Subsection \ref{subsec:counting}.

We use there the notion of {\em consistent values}, Definition \ref{def:cons}, as an algebraic characterization of when is a configuration of cells feasible for a given affine structure. This allows us to regulate ``all-or-nothing'' lemmas from previous works in Theorem \ref{thm:density}, to provide a calculated bound for the number of structures. We also utilize it for Lemma \ref{lem:present}, showing that the subcell selection process does not ``spoil'' a good configuration.

\subsection{Wrapping Up}

There are some final ingredients that we need before finalizing the proof of Theorem \ref{thm:main2}. One of which is a compactness argument, analogous to the one made in \cite{AS08}, to be able to bound the size of the constraints we need to test for, even when the property is defined by an infinite number of constraints. In our case, we also need to perform a slight ``preprocessing'' to representation of the property, to make it {\em concise} as per Definition \ref{def:concise}, which is done through Lemma \ref{lem:conc}. Apart from this, Subsection \ref{subsec:more} contains a few other algebraic tools that help with the calculations used in the proof.

Finally, Subsection \ref{subsec:proof} contains the proof of Theorem \ref{thm:main2}, tying it all together, from finding a factor with a subcell selection, through consistency and randomness arguments, to finally using the function cleanup to bound from below the number of copies of the corresponding induced structure.

\section{Tools of the Proof}\label{sec:tools}

In this section we lay the groundwork for the decomposition theorems that follow. This include the formal definition of partition by polynomial factors, the definition of factor robustness and rank with proofs of their impact, and finally we prove the main lemma about the existence of partitions that are both robust and of high rank.

\subsection{Functions and Norms}\label{subsec:norms}

In the most general setting we consider functions $f: G \to \C$, where $G$ is a finite Abelian group\footnote{Later we would mostly consider $G=\F_p^n$. Our main theorem is formulated for functions whose range is $\{0,1\}$, but its proof uses interim function with larger ranges.}.

Unless stated otherwise, expectations are taken over the uniform probability space with respect to the relevant range, e.g.\ $\E_x[f(x)]$ is set to $|G|^{-1}\sum_{x\in G}f(x)$. Apart from the traditional norms such as $\|f\|_2^2=\E_x[|f(x)|^2]$, we will make extensive use of Gowers norms.

\begin{definition}[Gowers norm]
Let $G$ be a finite Abelian group and $f: G \to \C$. For an integer $k
\geq 1$, the $k$'th Gowers norm of $f$, denoted $\|f\|_{U^k}$, is
defined by:
\begin{equation*}
\|f\|_{U^k}^{2^k} = \E_{x, y_1, y_2, \dots, y_k \in G}\left[\prod_{S \subseteq [k]}\mathcal{C}^{k-|S|}
  f\left(x + \sum_{i \in S} y_i\right)\right]
\end{equation*}
where $\mathcal{C}$ denotes the complex conjugation operator, i.e.\ $\mathcal{C}^l(a+bi)=a+(-1)^lbi$ for $a,b\in\R$ and integer $l$.
\end{definition}

Two facts about the Gowers norm will be absolutely crucial in what
follows. First is the Gowers Inverse theorem, established by
\cite{BTZ10, TZ}. Throughout, we let $\expo{x}$ denote the
complex number $e^{2\pi i x/p}$ for $x \in \F_p$.

\begin{theorem}[Gowers Inverse Theorem]\label{thm:git}
Given positive integers $d < p$,
for every $\delta > 0$, there exists $\eps =
\eps_{\ref{thm:git}}(\delta,p)$ such that if $f: \F_p^n \to \R$
satisfies $\|f\|_\infty\leq 1$ and $\|f\|_{U^{d+1}} \geq \delta$, then there
exists a polynomial $P: \F_p^n \to \F_p$ of degree at most $d$ so that
$|\E_x[f(x)\cdot \expo{P(x)}]| \geq \eps$.
\end{theorem}

The second is a lemma due to Green and Tao \cite{GT06} based on
repeated applications of the Cauchy-Schwarz inequality. Refer to
Definition \ref{def:cplx} for the term ``complexity''.
\begin{lemma}\label{lem:cnt}
Let $f_1, \dots, f_m : \F_p^n \to [-1,1]$. Let $\mathcal{L} =
\{L_1,\dots,L_m\}$ be a system of $m$ linear forms in $\ell$ variables
of complexity $s$. Then:
$$\left|\E_{x_1,\dots,x_\ell \in \F_p^n} \left[ \prod_{i=1}^m
    f_i(L_i(x_1,\dots,x_\ell))\right] \right| \leq \min_{i \in [m]}\|f_i\|_{U^{s+1}}$$
\end{lemma}

\subsection{Polynomial Factors and their Rank}\label{subsec:poly}

While partitioning the domain  to affine linear subspaces would be the
most intuitive for counting affine cubes, we in fact need higher degree algebraic partitions.

\begin{definition}[Polynomial factor] \label{def:factor} A {\em polynomial factor} $\calb$ is a sequence of
polynomials $P_1, \dots, P_C:\F_p^n \to \F_p$. We also identify it with the
function $\calb:\F_p^n\to\F_p^C$ sending $x$ to
$(P_1(x),\ldots,P_C(x))$.
A {\em cell} of $\calb$ is a preimage $\calb^{-1}(y)$ for some
$y\in\F_p^C$. On the other hand, given a cell of $\calb$, the common value
$y=\calb(x)\in\F_p^C$ is called the {\em image} of the cell. When there is no ambiguity,
we will in fact abuse notation and identify a cell of $\calb$ with its image $y$.

The {\em partition induced by $\calb$} is the
partition of $\F_p^n$ given by $\left\{\calb^{-1}(y):y\in \F_p^C\right\}$.
The {\em complexity} of $\calb$ is the number of defining polynomials
$|\calb|=C$. The {\em degree} of $\calb$ is the maximum degree among
its defining polynomials $P_1,\ldots,P_C$.
\end{definition}

Next, we define the notion of conditional expectation with respect to
a given factor.

\begin{definition}[Expectation over polynomial factor]
Given a factor $\calb$ and a function $f:\F_p^n\to\zo$, the
{\em expectation} of $f$ over a cell $y \in \F_p^{|\calb|}$ is the average
$\E_{x:\calb(x)=y}[f(x)]$, which we denote by $\E[f|y]$. The {\em
  conditional expectation} of $f$ over $\calb$, is the real-valued
function over $\F_p^n$ given by $\E[f|\calb](x)=\E[f|\calb(x)]$. In
particular, it is constant on every cell of the polynomial factor.
\end{definition}

In essence we would want to choose a polynomial factor so that, among other things, the restriction of $f$ in every cell would essentially consist of a constant element and other elements of small norms. However, since we are not dealing with affine linear subspaces, for our arguments to follow we also need the factor itself to be ``well behaved''. This is exemplified in the notion of polynomial rank \cite{GT07}, in essence a strengthening of linear independence.

\begin{definition}[Rank of polynomial factors]
Suppose that $\calb$ is a polynomial factor defined by
polynomials $P_1, \dots, P_C: \F_p^n \to \F_p$. The {\em rank of
  $\calb$} is the largest integer $r$ such that for every $(\alpha_1,
\dots, \alpha_C) \in \F_p^C \setminus \{0^C\}$, the polynomial
$P_\alpha = \sum_{i=1}^C \alpha_i P_i$ cannot be expressed as a
function of $r$ polynomials of degree $d-1$, where $d = \max_{i\in
  [C]: \alpha_i \neq 0} \deg(P_i)$.

The rank of a single polynomial $P$ is defined similarly (but without needing to relate to linear combinations).
\end{definition}

The following result, proved by Kaufman and Lovett
\cite{KL08} for all $p$ (extending previous work of Green and Tao
\cite{GT06} over large characteristic fields), is crucial:
\begin{theorem}\label{thm:rankreg}
For any $\eps > 0$ and integer $d \geq 1$, there exists
$r = r_{\ref{thm:rankreg}}(d,\eps)$ such that:  If $P: \F_p^n \to \F_p$ is a
degree-$d$ polynomial with rank at least $r$, then $|\E_x[\expo{P(x)}]|
< \eps$.
\end{theorem}
As an example of how useful Theorem \ref{thm:rankreg} is, consider the
following simple lemma which states that every cell of a polynomial
factor with large enough rank has approximately the same size.
\begin{lemma}\label{lem:cellsize}
Given a polynomial factor $\calb$ of degree $d$, complexity $C$, and
rank at least $r_{\ref{thm:rankreg}}(d,\eps)$ generated by the
polynomials $P_1,\dots,P_C: \F_p^n \to \F_p$, and an element $b \in
\F_p^C$, we have that:
\begin{equation*}
\Pr_{x \in \F_p^n}[\calb(x) = b] = p^{-C} \pm \eps
\end{equation*}
\end{lemma}
\begin{proof}
This is implicit in previous work, e.g. \cite{Gre07}. For
completeness, we repeat the argument:
\begin{align*}
\Pr_{x \in \F_p^n}[\calb(x) = b]
&= \E_{x} \left[ \prod_{i \in [C]}\frac{1}{p} \sum_{\lambda_i \in
    \F_p} \expo{\lambda_i\cdot (P_i(x) - b_i)}\right]\\
&= p^{-C} \sum_{(\lambda_1,\ldots,\lambda_C) \in \F_p^C}
\E_{x}\left[\expo{\sum_{i \in [C]} \lambda_i(P_i(x) - b_i)}\right]\\
&= p^{-C}\left(1 \pm p^C\eps \right)
\end{align*}
where the last line uses Theorem \ref{thm:rankreg} whenever $(\lambda_1,\ldots,\lambda_C)\neq 0^C$.
\end{proof}

\subsection{Refinement and Robustness}\label{subsec:refrob}

The decomposition theorems will iteratively partition the domain
$\F_p^n$ into finer and finer partitions (though we will use a
mechanism that hides the refinements that do not have to be
``visible'' for the other proofs). We will need to be
careful about distinguishing between two different types of
refinements.

\begin{definition}[Refinement of a polynomial factor] \label{def:refine}
$\calb'$ is called a {\em syntactic refinement} of $\calb$, and
denoted $\calb' \preceq_{syn} \calb$, if the sequence of polynomials
defining $\calb'$ extends that of $\calb$. It is called a {\em
  semantic refinement}, and denoted $\calb' \preceq_{sem} \calb$ if the
induced partition is a combinatorial refinement of the partition
induced by $\calb$. In other words, if for every $x,y\in \F_2^n$,
$\calb'(x)=\calb'(y)$ implies $\calb(x)=\calb(y)$. The relation $\preceq$
(without subscripts) is a synonym for $\preceq_{syn}$.
\end{definition}
Clearly, being a syntactic refinement is stronger than being a
semantic refinement. However in essence, these are almost the same
thing.

\begin{observation}\label{obs:silly}
If $\calb'$ is a semantic refinement of $\calb$, then there exists a
syntactic refinement $\calb''$ of $\calb$ that induces the same
partition of $\F_p^n$, and for which $|\calb''|\leq
|\calb'|+|\calb|$.
\end{observation}

\begin{proof}
Just add the defining polynomials of $\calb$ to those of $\calb'$.
\end{proof}

On the other hand, doing the above conversion can ``destroy'' the rank of a polynomial factor, and there will be indeed situations in what follows where we will have to carefully distinguish the two refinement types.

Next, we define the density index of a polynomial factor with respect to a function, and use it to define the notion of robustness, which is central to what follows.

\begin{definition}\label{def:indd}
The {\em density index} of a factor $\calb$ with respect to a function $f$ is the squared $l_2$ norm of the conditional expectation of $f$, that is $\indd(\calb)=\E\left[(\E[f|\calb])^2\right]$.

Given a function $h:\N\to\N$ and a real parameter $\gamma$, A factor $\calb$ is {\em $(h,\gamma)$-robust} (semantically) if there exists no $\calb'$ which is a semantic refinement of $\calb$ for which $|\calb'|\leq h(|\calb|)$ and $\indd(\calb')\geq\indd(\calb)+\gamma$.
\end{definition}

Robustness is somewhat preserved when moving to a small refinement.

\begin{observation}\label{obs:robspr}
If $\calb$ is $(g\circ h,\gamma)$-robust, and $\calb'$ is a (syntactic or semantic) refinement of $\calb$ for which $|\calb'|\leq h(|\calb|)$, then $\calb'$ is $(g,\gamma)$-robust.
\end{observation}

\begin{proof}
If $\calb''$ is any refinement of $\calb'$ for which $|\calb''|\leq g(|\calb'|)$, then $|\calb''|\leq g(h(|\calb|))$ and so $\indd(\calb'')\leq\indd(\calb)+\gamma$. On the other hand by the Cauchy-Schwarz inequality $\indd(\calb')\geq\indd(\calb)$, and so $\indd(\calb'')\leq\indd(\calb')+\gamma$, proving the robustness condition of $\calb'$.
\end{proof}

Existence of robust factors, also as syntactic refinements of a given factor, is easy to prove. Note that the function in its statement takes another function as one of its parameters.

\begin{observation}\label{obs:rob}
For an appropriate function $T_{\ref{obs:rob}}(k,h,\gamma)$, for any $\calb$, $h:\N\to\N$ and $\gamma>0$ there exists a syntactic refinement $\calb'$ which is $(h,\gamma)$-robust, and for which $|\calb'|\leq T_{\ref{obs:rob}}(|\calb|,h,\gamma)$.
\end{observation}

\begin{proof}
Without loss of generality we assume that $h$ is monotone non-decreasing (otherwise replace $h(k)$ with $\max_{j\leq k}h(j)$). Set $\calb_0=\calb$. Inductively, if $\calb_i$ is not already $(h,\gamma)$-robust then set $\calb'_i$ to be a semantic refinement of $\calb_i$ for which $|\calb'_i|\leq h(|\calb_i|)$ and $\indd(\calb')\geq\indd(\calb)+\gamma$, and by Observation \ref{obs:silly} then set $\calb_{i+1}$ to be a syntactic refinement of $\calb$ and $\calb'_i$ for which $|\calb_{i+1}|\leq h(|\calb_i|)+|\calb|$.

Noting that the index can only increase while moving to a refinement (by the Cauchy-Schwarz inequality), this process must stop for some $j\leq 1/\gamma$. $\calb_j$ is the required factor, and its size is bounded by $k_{1/\gamma}$, where we define $k_0=k$ and by induction $k_{i+1}=h(k_i)+k$.
\end{proof}

{\bf Note:} From now on we assume that all our relevant functions are monotone in their corresponding variables, also when this is not stated explicitly. For example, a function $h$ fed to Observation \ref{obs:rob} will assumed to be monotone non-decreasing, and if $k\leq k'$, $\gamma\geq\gamma'$, and $h(m)\leq h'(m)$ for every $m\in\N$ (while both $h$ and $h'$ are monotone non-decreasing), then $T_{\ref{obs:rob}}(k,h,\gamma)\leq T_{\ref{obs:rob}}(k',h',\gamma')$. All our lemmas can indeed be made to provide such functions.

\subsection{Robustness with Rank}\label{subsec:robrank}

The next item on the agenda is to show that polynomial factors can be refined to ones of high rank. The following index definition is used for analyzing rank.

\begin{definition}
The {\em degree index} of a factor $\calb$ is the (infinite but almost everywhere zero) sequence of non-negative integers $\indm(\calb)=I=(i_1,i_2,\ldots)$, where $i_k$ is the number of polynomials of degree $k$ in the sequence of polynomials defining $\calb$.

Denote the set of all possible degree sequences as above by $\cali$. Over $\cali$ we define the {\em anti-lexicographic} order, where $I<I'$ if $i_k<i'_k$ for the largest $k$ on which those coordinates differ.
\end{definition}

The set $\cali$ defined above is well-ordered in the sense that there exist no infinite strictly decreasing sequences of members of $\cali$, but this still does not provide for ``standard'' induction, as the order is not isomorphic to $\N$. To replace induction we define the notion of a decrement.

\begin{definition}
Let $\cali$ denote the well-ordered set of all possible degree indexes. A function $\kappa:\N\times\cali\to\cali$ is called a {\em decrement} if for all $A\in\cali$ and $n\in\N$ it satisfies $\kappa(n,A)<A$, for all $n$ and $A\leq B$ it satisfies $\kappa(n,A)\leq\kappa(n,B)$, and for all $n<m$ and $A$ it satisfies $\kappa(n,A)\leq\kappa(m,A)$. The inequalities are with respect to the anti-lexicographic ordering of $\cali$.
\end{definition}

The following shows how, when we are given a decrement that ``bounds'' some process, we can use it to bound an iterative process.

\begin{lemma}\label{lem:decbound}
There exist $T_{\ref{lem:decbound}}(k,d,h,\kappa)$ and $m_{\ref{lem:decbound}}(k,d,h,\kappa)$ that take numbers $k$ and $d$, a monotone $h:\N\to\N$ and a decrement $\kappa$, and satisfy the following. If $\calb_0,\calb_1,\ldots\calb_m$ is a sequence of factors of bounded degree $d$ for which $|\calb_0|\leq k$,  $|\calb_i|\leq h(|\calb_{i-1}|)$ and $\indm(\calb_i)\leq \kappa\left(|\calb_{i-1}|,\indm(\calb_{i-1})\right)$, then $|\calb_m|$ is bounded by $T_{\ref{lem:decbound}}(|\calb|,d,h,\kappa)$ and $m$ is bounded by $m_{\ref{lem:decbound}}(|\calb|,d,h,\kappa)$.
\end{lemma}

\begin{proof}
Let $I_0$ be the maximal (with respect to order) degree index of any degree $d$ factor of complexity $k$, which is the sequence $(i_1,i_2,\ldots)$ for which $i_d=k$ and $i_j=0$ for any $j\neq d$, and let $h_0=k$. Inductively define $h_i\in\N$ as $h_i=h(h_{i-1})$, and $I_i=\kappa(h_{i-1},I_{i-1})$. Because $I_0,I_1,\ldots$ is a decreasing sequence over a well-ordered set, it must be of bounded length, which we denote as $m_{\ref{lem:decbound}}(k,d,h,\kappa)$. We then set $T_{\ref{lem:decbound}}(k,d,h,\kappa)=h_{m_{\ref{lem:decbound}}(k,d,h,\kappa)}$. For a sequence of factors as above, the monotonicity conditions of $\kappa$ ensure that $|\calb_i|\leq h_i$ and $\indm(\calb_i)\leq I_i$, and so we are done.
\end{proof}

The following provides a decrement that will bound the process of obtaining a high rank refinement of a given factor, as well as a bound on the size increment. Note that also if the required rank depends on the factor size, we can still get a bounding decrement.

\begin{lemma}\label{lem:robrank}
For every $r:\N\to\N$ there exist $h_{\ref{lem:robrank}}^{(r)}:\N\to\N$ and a decrement $\kappa_{\ref{lem:robrank}}^{(r)}:\N\times\cali\to\cali$, satisfying the following for every $d$. If $\calb$ is a factor of degree at most $d$ whose rank is less than $r(|\calb|)$, then there exists a semantic refinement $\calb'$ of $\calb$ for which $|\calb'|\leq h_{\ref{lem:robrank}}^{(r)}(|\calb|)$ and $\indm(\calb')\leq\kappa_{\ref{lem:robrank}}^{(r)}(|\calb|,\indm(\calb))$.

Moreover, if $\calb$ is in itself a syntactic refinement of some $\hat{\calb}$ that is of rank at least $r(|\calb|)+1$, then additionally $\calb'$ will be a syntactic refinement of $\hat{\calb}$.
\end{lemma}

\begin{proof}
We will deal with the first case, and then show how to modify the proof for the case where being a syntactic refinement of some $\hat{\calb}$ of the appropriate rank must be preserved.

Let $p_1,\ldots,p_C$ be the defining polynomials for $\calb$, where $C=|\calb|$. Suppose there is a linear combination over $\F$ that shows that $\calb$ has a rank smaller than $r(C)$. This means that for some $(\alpha_1,\ldots,\alpha_C)\in\F^C\setminus\{0^C\}$, some arbitrary function $B:\F^l\to\F$ and polynomials $q_1,\ldots,q_l$ we have $\sum_{j=1}^C\alpha_jp_j(x)=B(q_1(x),\ldots,q_l(x))$ for every $x\in\F^n$, where $l<r(C)$ and every $q_i$ is of degree smaller than $\max\{\deg(p_j)|\alpha_j\neq 0\}$ (a possible special case is where $l=0$ and $B$ is a constant).

We select $j_0$ so that $\alpha_{j_0}\neq 0$ and $\deg(p_{j_0})=\max\{\deg(p_j)|\alpha_j\neq 0\}$, and construct $\calb'$ by replacing $p_{j_0}$ with $q_1,\ldots,q_l$. This is clearly a semantic refinement of $\calb$ of complexity bounded by $h(C)=C+r(C)-1$. Also, if $I=(i_1,\ldots)$ was the degree index of $\calb$, then the degree index of $\calb'$ is bounded above by the following $\kappa(C,I)=(j_1,\ldots)$: Letting $k$ be the smallest number such that $i_k>0$, we set $j_k=i_k-1$, and if $k>1$ then we set $j_{k-1}=i_{k-1}+r(C)-1$; all other coordinates of $\kappa(C,I)$ are set equal to the respective coordinates of $I$.

The above argument provides us with $h_{\ref{lem:robrank}}^{(r)}$ and $\kappa_{\ref{lem:robrank}}^{(r)}$ as required.

Now we deal with an existing $\hat{\calb}$ as above. We follow the same argument, but argue that we can find $j_0$ for which $\alpha_{j_0}\neq 0$ that corresponds to a maximal degree polynomial, satisfying additionally $j_0>\hat{C}=|\hat{\calb}|$. Assuming otherwise, we would find a counter example to the rank assumption on $\hat{B}$: We would get that $\sum_{j=1}^{\hat{C}}\alpha_jp_j(x)$ can be expressed as a function of $q_1,\ldots,q_l$ and $q_{l+1}=\sum_{j=\hat{C}+1}^C\alpha_jp_j$, which would all be of lower degree than $\max\{\deg(p_j)|1\leq j\leq \hat{C}, \alpha_j\neq 0\}=\max\{\deg(p_j)|\alpha_j\neq 0\}$, and would hence violate the rank of $\hat{\calb}$.
\end{proof}

Now we can combine the above two lemmas and prove the existence of high rank refinements.

\begin{lemma}\label{lem:rank}
There exists $D_{\ref{lem:rank}}^{(d,r)}(k)$ which takes two numbers $k$ and $d$ and a monotone function $r:\N\to\N$, and satisfies the following. For every factor $\calb$ of bounded degree $d$ there is a semantic refinement $\calb'$ for which $|\calb'|\leq D_{\ref{lem:rank}}^{(d,r)}(|\calb|)$, is of bounded degree $d$ and has rank at least $r(|\calb'|)$.

Moreover, if $\calb$ is in itself a syntactic refinement of some $\hat{\calb}$ that is of rank at least $r(D_{\ref{lem:rank}}^{(d,r)}(|\calb|))+1$, then additionally $\calb'$ will be a syntactic refinement of $\hat{\calb}$.
\end{lemma}

\begin{proof}
We set $D_{\ref{lem:rank}}^{(d,r)}(k)=T_{\ref{lem:decbound}}(k,d,h_{\ref{lem:robrank}}^{(r)},\kappa_{\ref{lem:robrank}}^{(r)})$. We set $\calb_0=\calb$, and as long as $\calb_i$ is of rank less than $r(|\calb_i|)$ we move to a semantic refinement $\calb_{i+1}$ as guaranteed by Lemma \ref{lem:robrank}. By Lemma \ref{lem:decbound} the sequence $\calb_0,\calb_1,\ldots$ has length bounded by $m_{\ref{lem:decbound}}(k,d,h_{\ref{lem:robrank}}^{(r)},\kappa_{\ref{lem:robrank}}^{(r)})$, and the final factor $\calb_l$ is of rank at least $r(|\calb_l|)$ (otherwise we could have continued the sequence) and of complexity bounded by $T_{\ref{lem:decbound}}(k,d,h_{\ref{lem:robrank}}^{(r)},\alpha_{\ref{lem:robrank}}^{(r)})$.

For the case of a prior factor $\hat{\calb}$ we just use the corresponding case of Lemma \ref{lem:robrank}.
\end{proof}

Now we finally state the main technical lemma that we will use for our decompositions. It will find a refinement that is both robust and of high rank, while not breaking a given syntactic refinement relation to a high rank factor if one exists.

\begin{lemma}[main robustness lemma]\label{lem:rankrob}
For an appropriate function $T_{\ref{lem:rankrob}}(k,h,d,r,\gamma)$, for any $\calb$ of degree bound $d$, monotone $h:\N\to\N$ and $r:\N\to\N$, and $\gamma\in(0,1)$ there exists a semantic refinement $\calb'$ of $\calb$ which is of rank at least $r(|\calb'|)$ and $(h,\gamma)$-robust, for which $|\calb'|\leq T_{\ref{lem:rankrob}}(|\calb|,h,d,r,\gamma)$.

Moreover, if $\calb$ is in itself a syntactic refinement of some $\hat{\calb}$ that is of rank at least $r(T_{\ref{lem:rankrob}}(|\calb|,h,d,r,\gamma))+1$, then additionally $\calb'$ will be a syntactic refinement of $\hat{\calb}$ (this holds also for the case where $\calb=\hat{\calb}$).
\end{lemma}

\begin{proof}
We set $T_{\ref{lem:rankrob}}(k,h,d,r,\gamma)=D_{\ref{lem:rank}}^{(d,r)}\left(T_{\ref{obs:rob}}(k,h\circ D_{\ref{lem:rank}}^{(d,r)},\gamma)\right)$. Given $\calb$, we first use Lemma \ref{obs:rob} to find $\calb_1$ that is a syntactic refinement of $\calb$ and is $(h\circ D_{\ref{lem:rank}}^{(d,r)},\gamma)$-robust. We then let $\calb'$ be its semantic refinement according to Lemma \ref{lem:rank} that is of rank $r(|\calb'|)$. The complexity of $\calb'$ is at most $D_{\ref{lem:rank}}^{(d,r)}(|\calb_1|)$, and hence (apart from being bounded by the above $T_{\ref{lem:rankrob}}(|\calb|,h,d,r,\gamma)$) by Observation \ref{obs:robspr} it is $(h,\gamma)$-robust as required.

For the case where there is a prior factor $\hat{\calb}$ of the stated rank, we just use the corresponding case of Lemma \ref{lem:rank}.
\end{proof}

\section{Decomposition Theorems}\label{sec:decomp}

We use here the tools of the previous section to prove two decomposition theorems. First we state and prove the strong decomposition theorem (it is called ``strong'' on account of also guaranteeing high rank); similar theorems were proved in previous works, and we only make a seemingly small (yet crucial to what follows) addition that preserves a given syntactic refinement relation. Then we state and prove the super decomposition theorem, which uses the strong decomposition theorem (or more accurately the main lemma implying it) as a lemma.

Super decomposition provides us with two successive factors, one being a syntactic refinement of the other. For the testing proofs, instead of using it directly, we will use a corollary that ``chooses'' out of the finer factor only one representative for each of the cells of coarser factor. This is done in the subcell selection corollary. The resulting representatives will satisfy properties that are stronger than what {\em any} one factor can satisfy by itself.

\subsection{Strong Decomposition}\label{subsec:strong}
First, a corollary of Theorem \ref{thm:git}.

\begin{lemma}\label{lem:inc}
For $d<p$, suppose that $\calb$ is a polynomial factor of degree $d$ and complexity
$C$, and suppose $f: \F_p^n \to \zo$ is such that
$\|f-\E[f|\calb]\|_{U^{d+1}} \geq \delta$. Then, there exists a refined
polynomial factor $\calb'$ of degree $d$ and complexity at most $C+1$
such that:
$$\|\E[f|\calb']\|_2^2 \geq \|\E[f|\calb]\|_2^2 +
(\eps_{\ref{thm:git}}(\delta,p))^2$$
where $\eps_{\ref{thm:git}}$ is the function in Theorem \ref{thm:git}.
\end{lemma}
\begin{proof}
$g= f-\E[f|\calb]$ is bounded to $[-1,1]$. So, applying Theorem
\ref{thm:git} yields a degree-$d$ polynomial $P$ satisfying
$|\E[g(x) \cdot \expo{P(x)}]| \geq \eps_{\ref{thm:git}}(\delta,p)$. The
polynomial $P$ generates a factor $\hat{\calb}$ of complexity
$1$. Define $\calb'$ to be the common refinement of $\calb$ and
$\hat{\calb}$ (by adding $P$ to the polynomials defining $\calb$); its complexity is $C+1$.

Observe that:
\begin{align*}
&\|\E[g|\calb']\|_1 = \E_x\left[|\E[g|\calb'](x)|\right] =
\E_x\left[|\E[g|\calb'](x)\cdot \expo{P(x)}|\right] \\
& \geq \left| \E_x\left[\E[g|\calb'](x)\cdot \expo{P(x)}\right] \right| =
\left|\E_x\left[g(x) \cdot
\expo{P(x)}\right]\right| \geq \eps_{\ref{thm:git}}(\delta,p)
\end{align*}
where the second equality is simply due to $|\expo{P(x)}|=1$, and the third equality uses the fact that $P$ is constant on each atom of $\calb'$.
Now finally:
$$\|\E[f|\calb']\|_2^2 - \|\E[f|\calb]\|_2^2 = \|\E[f|\calb'] -
\E[f|\calb]\|_2^2 = \|\E[g|\calb']\|_2^2 \geq \|\E[g|\calb']\|_1^2
\geq \eps_{\ref{thm:git}}^2(\delta,p)$$
where the first equality uses the fact that $\calb'$ is a refinement of $\calb$.
\end{proof}

The contra-positive of the above provides us with a function decomposition given a sufficiently robust polynomial factor.

\begin{lemma}\label{lem:robdec1}
For any $\eta$ and $d<p$ there exist $h_{\ref{lem:robdec1}}:\N\to\N$ and $\gamma_{\ref{lem:robdec1}}(\eta,p)$, so that if $\calb$ is $(h_{\ref{lem:robdec1}},\gamma_{\ref{lem:robdec1}}(\eta,p))$-robust (with respect to $f$) among factors of degree bound $d$ over $\F_p^n$, then there is a decomposition $f=f_1+f_2$ where $f_1$ is constant over every atom of $\calb$ and ranges in $[0,1]$, and $f_2$ satisfies $\|f_2\|_{U^k} \leq \eta$ and ranges in $[-1,1]$.
\end{lemma}

\begin{proof}
We set simply $h_{\ref{lem:robdec1}}(k)=k+1$ and $\gamma_{\ref{lem:robdec1}}(\eta,p)=\eps_{\ref{thm:git}}(\eta,p)^2$. Given $\calb$ as above we set $f_1=\E[f|\calb]$ and $f_2=f-\E[f|\calb]$. These functions clearly have the required ranges. The robustness condition of $\calb$ implies the contra-positive of the conclusion of Lemma \ref{lem:inc}, and so we must have $\|f_2\|_{U^k} \leq \eta$ as required.
\end{proof}

However, we would like to make the Gowers norm bound also a function of $|\calb|$. For this we will decompose $f$ into three functions, where the third ``error term'' function has a bound on its $l_2$ norm. In fact an $l_2$ norm bound is what we need for an error term, but to reach even a constant $l_2$ norm bound we cannot avoid having also the function that has ``only'' a Gowers norm bound.

\begin{lemma}\label{lem:robdec2}
For any $d<p$,  $\delta$ and $\eta:\N\to \R^+$ there exist $h_{\ref{lem:robdec2}}^{(\eta,p)}:\N\to\N$ and $\gamma_{\ref{lem:robdec2}}(\delta)$, so that if $\calb$ is $(h_{\ref{lem:robdec2}}^{(\eta,p)},\gamma_{\ref{lem:robdec2}}(\delta))$-robust (with respect to $f$) among factors of degree bound $d$, then there is a decomposition $f=f_1+f_2+f_3$ where $f_1$ is constant over every atom of $\calb$ and ranges in $[0,1]$, $f_2$ satisfies $\|f_2\|_{U^k} \leq \eta(|\calb|)$ and ranges in $[-1,1]$, and $f_3$ ranges in $[-1,1]$ and satisfies $\|f_3\|_2\leq\delta$, where $f_1+f_3$ also ranges in $[0,1]$.
\end{lemma}

\begin{proof}
We set $h_{\ref{lem:robdec2}}^{(\eta,p)}(m)=T_{\ref{obs:rob}}\left(m,h_{\ref{lem:robdec1}},\gamma_{\ref{lem:robdec1}}(\eta(m),p)\right)$ for every $m\in\N$ and $\gamma_{\ref{lem:robdec2}}(\delta)=\delta^2$. Given $\calb$ satisfying the robustness condition above, we let $\calb'$ be its syntactic refinement which is $\left(h_{\ref{lem:robdec1}},\gamma_{\ref{lem:robdec1}}(\eta(|\calb|),p)\right)$-robust and for which $|\calb'|\leq T_{\ref{obs:rob}}\left(|\calb|,h_{\ref{lem:robdec1}},\gamma_{\ref{lem:robdec1}}(\eta(|\calb|),p)\right)$. We let $f_1=\E[f|\calb]$, and $f_2=f-\E[f|\calb']$. As per Lemma \ref{lem:robdec1} $f_2$ satisfies the required Gowers norm condition. This leaves us with $f_3=\E[f|\calb']-\E[f|\calb]$. The required $l_2$ condition on this function follows directly from $\calb'$ not violating the robustness condition of $\calb$.
\end{proof}

We now have all the tools to quickly wrap up the proof of the existence of a strong decomposition.

\begin{theorem}[Strong Decomposition Theorem]\label{thm:strongdecomp}
Suppose $\delta > 0$ and $C_0,d \geq 1$ are integers so that
$d<p$. Let $\eta: \N \to \R^+$ be an arbitrary
non-increasing function and $r: \N \to \N$ be an arbitrary
non-decreasing function. Then there exists $C =
C_{\ref{thm:strongdecomp}}(\delta,\eta,p,r, C_0)$ such that the following holds.

Given $f: \F_p^n \to \zo$ and a polynomial factor
$\calb_0$ of degree at most $d$ and complexity at most $C_0$, there
exist three functions $f_1, f_2, f_3: \F_p^n \to
\R$ and a polynomial factor  $\calb \preceq_{sem} \calb_0$ of
degree at most $d$ and complexity at most $C$ such that the following hold:
\begin{itemize}
\item
$f=f_1+f_2+f_3$
\item
$f_1 = \E[f|\calb]$
\item
$\|f_2\|_{U^{d+1}} \leq 1/\eta(|\calb|)$
\item
$\|f_3\|_2 \leq \delta$
\item
$f_1$ and $f_1 + f_3$ have range $[0,1]$; $f_2$ and $f_3$ have range $[-1,1]$.
\item
$\calb$ is of rank at least $r(|\calb|)$
\end{itemize}
Moreover, if $\calb_0$ is a syntactic refinement
of some $\hat{\calb}$ of rank at least $r(C)+1$, then $\calb$ will also
be a syntactic refinement of $\hat{\calb}$ (in particular this also holds if $\calb_0=\hat{\calb}$).
\end{theorem}

\begin{proof}
Set $C_{\ref{thm:strongdecomp}}(\delta,\eta,p,r,C_0)=T_{\ref{lem:rankrob}}(C_0,h_{\ref{lem:robdec2}}^{(\eta,p)},p,r,\gamma_{\ref{lem:robdec2}}(\delta))\geq T_{\ref{lem:rankrob}}(C_0,h_{\ref{lem:robdec2}}^{(\eta,p)},d,r,\gamma_{\ref{lem:robdec2}}(\delta))$.
Given $\calb_0$ and $f$, we set $\calb$ to be the $(h_{\ref{lem:robdec2}}^{(\eta,p)},\gamma_{\ref{lem:robdec2}}(\delta))$-robust refinement of $\calb_0$ guaranteed by Lemma \ref{lem:rankrob}. Lemma \ref{lem:robdec2} guarantees the required decomposition $f=f_1+f_2+f_3$, and the case of a prior $\hat{\calb}$ is also handled seamlessly by Lemma \ref{lem:rankrob}.
\end{proof}

\subsection{Super Decomposition and Subcell Selection}\label{subsec:super}

What we would really like is that in some sense the $\delta$ of Theorem \ref{thm:strongdecomp} would also be able to depend on $|\calb|$, but this is clearly impossible. So instead, taking some inspiration from \cite{AFKS}, we will strive to have a sequence of two factors $\calb$ and $\calb'$, the latter a syntactic refinement of the former, so that the $\delta$ of $\calb'$ would be a function of $|\calb|$. However, for this to mean anything we also need $\calb'$ to ``faithfully'' represent $\calb$, in the sense that we define now.

\begin{definition}[Polynomial factor represents another factor]
Given a function $f: \F_p^n \to \zo$, a polynomial factor $\calb'$
that syntactically refines another factor $\calb$ and a real $\zeta \in (0,1)$, we
say {\em $\calb'$ $\zeta$-represents $\calb$ with respect to $f$} if
for at most a $\zeta$ fraction of cells $c$ of $\calb$, more than
$\zeta$ fraction of the cells $c'$ lying inside $c$ satisfy
$|\E[f|c]-\E[f|c']|>\zeta$.
\end{definition}

To be able to infer that a refinement is representing, we will use the following well-known defect version of the Cauchy-Schwarz inequality:
\begin{observation}\label{obs:defect}
If $\sum_{i\in I}^m\alpha_i=1$ where $\alpha_i$ are all non-negative, $f:I\to\R$ ranges over $[0,1]$, and for some $J\subseteq I$ we have $\left(\sum_{j\in J}\alpha_if(i)\right)/\left(\sum_{j\in J}\alpha_i\right)=\sum_{i\in I}\alpha_if(i)+\eta$ where $\eta\in[-1,1]$, then $\sum_{i\in I}\alpha_i(f(i)^2)\geq(\sum_{i\in I}\alpha_if(i))^2+\left(\sum_{j\in J}\alpha_i\right)\eta^2$.
\end{observation}

\begin{proof}
For ease of notation denote the average $a=\sum_{i\in I}\alpha_if(i)$ of $f$ and set $\xi=\sum_{j\in J}\alpha_i$. By the standard Cauchy-Schwarz inequality $\sum_{i\in I}\alpha_i(f(i)^2)\geq\sum_{i\in I}\alpha_i(f'(i)^2)$, where $f'(i)=\left(\sum_{j\in J}\alpha_if(i)\right)/\left(\sum_{j\in J}\alpha_i\right)=a+\eta$ if $i\in J$ and $f'(i)=\left(\sum_{j\in I\setminus J}\alpha_if(i)\right)/\left(\sum_{j\in I\setminus J}\alpha_i\right)=a-\xi\eta/(1-\xi)$ if $i\not\in J$. The sum over $f'$ now equals $\xi(a+\eta)^2+(1-\xi)(a-\xi\eta/(1-\xi))^2\geq a^2+\xi\eta^2$.
\end{proof}

We can now show that, under some rank assumptions, a non-representing refinement is evidence to a factor being non-robust.

\begin{lemma}\label{lem:repadd}
There are functions $r_{\ref{lem:repadd}}(p,m)$ and $\gamma_{\ref{lem:repadd}}(\zeta)$ for which the following holds.
For $f: \F_p^n \to \zo$, if $\calb'$ is a factor of rank $r_{\ref{lem:repadd}}(p,|\calb'|)$, and is a syntactic refinement of a factor $\calb$ of rank $r_{\ref{lem:repadd}}(p,|\calb|)$, both of degree $d<p$, and $\calb'$ does not
$\zeta$-represent $\calb$ with respect to $f$, then $\indd(\calb')\geq\indd(\calb)+\gamma_{\ref{lem:repadd}}(\zeta)$.
\end{lemma}

\begin{proof}
We first set $r_{\ref{lem:repadd}}(p,m)=r_{\ref{thm:rankreg}}(p,1/2p^m)$. If $\calb'$ does not $\zeta$-represent $\calb$, then it must be the case that there are at least $\zeta p^{|\calb|}/2$ cells of $\calb$, so that for every cell $c$ of them, there are at least $\zeta p^{|\calb'|-|\calb|}/2$ cells $c'$ of $\calb'$ lying inside of it, so that $|\E[f|c]-\E[f|c']|>\zeta$.

Let us concentrate for now on one such cell $c$ of $\calb$. Either there are at least $\zeta p^{|\calb'|-|\calb|}/4$ cells $c'$ inside $c$ so that $\E[f|c']-\E[f|c]>\zeta$, or there are more than $\zeta p^{|\calb'|-|\calb|}/4$ such cells so that $\E[f|c']-\E[f|c]<-\zeta$. We will assume the first case, as the treatment of the second case is virtually identical and provides the same lower bound for the cell.

Now we refer to Observation \ref{obs:defect}, where $I$ is identified with $\F_p^{|\calb'|-|\calb|}$, the set of cells of $\calb'$ lying in $c$, and $J$ is identified with the set of those cells $c'$ satisfying $\E[f|c']-\E[f|c]>\zeta$. The value of each $\alpha_i$ can easily be shown to be at least $p^{|\calb'|-|\calb|}/3$, by comparing the minimum possible size of the cell $c'$ with the maximum possible size of the cell $c$. Inserting the other corresponding values in Observation \ref{obs:defect}, we obtain $\E[\E[(f(x))^2]|c]>(\E[(f(x))^2|c])^2+\zeta^3/12$.

Summing up the above contribution for all cells $c$ of $\calb$, and noting that the relative size of every cell of $\calb$ is at least $p^{-|\calb|}/2$ by Lemma \ref{lem:cellsize}, we obtain that $\indd(\calb')=\E[(f(x))^2|\calb']\geq\E[(f(x))^2|\calb]+\zeta^3/24=\indd(\calb)+\gamma_{\ref{lem:repadd}}(\zeta)$, where we set $\gamma_{\ref{lem:repadd}}(\zeta)=\zeta^3/24$.
\end{proof}

The following technical lemma shows that if the partition is robust enough, then it has a specified robust and representing {\em syntactic} refinement, where we also take a rank requirement into account.

\begin{lemma}\label{lem:prepsuper}
For every $h:\N\to\N$, $\gamma:\N\to(0,1)$, $r:\N\to\N$, $p\in\N$ and $\zeta\in(0,1)$ there are $H_{\ref{lem:prepsuper}}^{(h,\gamma,p,r)}:\N\to\N$, $R_{\ref{lem:prepsuper}}^{(h,\gamma,p,r)}:\N\to\N$ and $\Gamma_{\ref{lem:prepsuper}}(\zeta)\in(0,1)$ satisfying the following among factors of degree bound $d<p$ over $\F_p^n$. If $\calb$ is an $(H_{\ref{lem:prepsuper}}^{(h,\gamma,p,r)},\Gamma_{\ref{lem:prepsuper}}(\zeta))$-robust partition of rank at least $R_{\ref{lem:prepsuper}}^{(h,\gamma,p,r)}(|\calb|)$, then it has a $\zeta$-representing syntactic refinement $\calb'$ which is $(h,\gamma(|\calb|))$-robust and is of rank at least $r(|\calb'|)$, which satisfies also $|\calb'|\leq S_{\ref{lem:prepsuper}}(|\calb|,h,p,r,\gamma)$ for the appropriate function $S_{\ref{lem:prepsuper}}(m,h,p,r,\gamma)$.
\end{lemma}

\begin{proof}
Set the following in order:
\begin{eqnarray*}
S_{\ref{lem:prepsuper}}(m,h,p,r,\gamma) &=& T_{\ref{lem:rankrob}}(m,h,p,r,\gamma(m))\\
H_{\ref{lem:prepsuper}}^{(h,\gamma,p,r)}(m) &=& S_{\ref{lem:prepsuper}}(m,h,p,r,\gamma)\\
R_{\ref{lem:prepsuper}}^{(h,\gamma,p,r)}(m) &=& \max\{r(S_{\ref{lem:prepsuper}}(m,h,p,r,\gamma))+1,r_{\ref{lem:repadd}}(p,m)\}\\
\Gamma_{\ref{lem:prepsuper}}(\zeta) &=& \gamma_{\ref{lem:repadd}}(\zeta)
\end{eqnarray*}
Assuming that $\calb$ satisfies the requisites, we use Lemma \ref{lem:rankrob} to find a refinement $\calb'$ that is $(h,\gamma(|\calb|))$-robust, of rank at least $r(|\calb'|)$, and satisfying $|\calb'|\leq T_{\ref{lem:rankrob}}(|\calb|,h,d,r,\gamma(|\calb|)) \leq T_{\ref{lem:rankrob}}(|\calb|,h,p,r,\gamma(|\calb|))$ -- the required complexity bound (note that Lemma \ref{lem:rankrob} is fed the number $\gamma(|\calb|)$, not the function $\gamma$).

The condition that $\calb$ is $(H_{\ref{lem:prepsuper}}^{(h,\gamma,p,r)},\Gamma_{\ref{lem:prepsuper}}(\zeta))$-semantically-robust means that $\indd(\calb')\leq\indd(\calb)+\gamma_{\ref{lem:repadd}}(\zeta)$, and so $\calb'$ is $\zeta$-representing for $\calb$ by Lemma \ref{lem:repadd} (as the partitions also satisfy the corresponding rank requirement).

The condition that $\calb$ is of rank at least $R_{\ref{lem:prepsuper}}^{(h,\gamma,p,r)}(|\calb|)\geq r\left(T_{\ref{lem:rankrob}}(|\calb|,h,p,r,\gamma(|\calb|))\right)+1$ means that (setting $\hat{\calb}=\calb$) Lemma \ref{lem:rankrob} provides the additional requirement that $\calb'$ is a syntactic refinement of $\calb$.
\end{proof}

We can now put forth our final decomposition theorem.

\begin{theorem}[Super Decomposition Theorem]\label{thm:superdecomp}
Suppose $\zeta > 0$ and $ d, C_0 \geq 1$ are integers so
that $d<p$. Let $\eta: \N \to \R^+$ and
$\delta:\N\to\R^+$ be arbitrary non-increasing functions, and
$r:\N\to\N$ be an arbitrary non-decreasing function. Then there exists $C =
C_{\ref{thm:superdecomp}}(\delta,\eta,p,r,\zeta, C_0)$ such that
the following holds.

Given $f: \F_p^n \to \zo$ and a polynomial factor
$\calb_0$ of degree at most $d$ and complexity at most $C_0$, there
exist functions $f_1, f_2, f_3: \F_p^n \to \R$, a semantic refinement
$\calb$ of $\calb_0$ of degree at most $d$ and a syntactic refinement
$\calb'$ of $\calb$ of degree at most $d$ and of complexity at most
$C$, such that the following hold:
\begin{itemize}
\item
$f=f_1+f_2+f_3$
\item
$f_1 = \E[f|\calb']$
\item
$\|f_2\|_{U^{d+1}} \leq \eta(|\calb'|)$
\item
$\|f_3\|_2 \leq \delta(|\calb|)$
\item
$f_1$ and $f_1 + f_3$ have range $[0,1]$; $f_2$ and $f_3$ have
range $[-1,1]$.
\item
$\calb$ is of rank at least $r(|\calb|)$.
\item
$\calb'$ is of rank at least $r(|\calb'|)$.
\item
$\calb'$ $\zeta$-represents $\calb$ with respect to $f$.
\end{itemize}
\end{theorem}

\begin{proof}
Let the function $\gamma$ be defined by $\gamma(m)=\gamma_{\ref{lem:robdec2}}(\delta(m))$ and let $h$ be defined by $h(m)=h_{\ref{lem:robdec2}}^{(\eta,p)}(m)$. Then set:

$$C_{\ref{thm:superdecomp}}(\delta,\eta,p,r,\zeta,C_0)= S_{\ref{lem:prepsuper}}\left(T_{\ref{lem:rankrob}}\left(C_0,H_{\ref{lem:prepsuper}}^{(h,\gamma,p,r)},p,R_{\ref{lem:prepsuper}}^{(h,\gamma,p,r)},\Gamma_{\ref{lem:prepsuper}}(\zeta)\right),h,p,r,\gamma\right).$$

Given $\calb_0$, we set $\calb$ to be the semantic refinement that is guaranteed by Lemma \ref{lem:rankrob} that is $\left(H_{\ref{lem:prepsuper}}^{(h,\gamma,p,r)},\Gamma_{\ref{lem:prepsuper}}(\zeta)\right)$-robust and is of rank at least $R_{\ref{lem:prepsuper}}^{(h,\gamma,p,r)}(|\calb|)$.
$|\calb|$ will be bounded by $T_{\ref{lem:rankrob}}\left(C_0,H_{\ref{lem:prepsuper}}^{(h,\gamma,p,r)},p,R_{\ref{lem:prepsuper}}^{(h,\gamma,p,r)},\Gamma_{\ref{lem:prepsuper}}(\zeta)\right)$.
Note also that $R_{\ref{lem:prepsuper}}^{(h,\gamma,p,r)}(|\calb|)\geq r(|\calb)$.

Now we can use Lemma \ref{lem:prepsuper} to provide us a $\zeta$-representing syntactic refinement $\calb'$ of $\calb$, that is of rank at least $r(|\calb'|)$, and is $(h,\gamma(|\calb|))$-robust and thus $\left(h_{\ref{lem:robdec2}}^{(\eta,p)},\gamma_{\ref{lem:robdec2}}(\delta(|\calb|))\right)$-robust. The factor $\calb'$ satisfies the required complexity upper bound by substituting the bound on $|\calb|$ into the guaranteed complexity bound of Lemma \ref{lem:prepsuper}. Finally Lemma \ref{lem:robdec2} provides the required decomposition $f=f_1+f_2+f_3$ over $\calb'$.
\end{proof}

One could envision future applications in which we would need the whole of $\calb'$. Here we will need a careful choice of one cell of $\calb'$ for every cell of $\calb$. This selection will satisfy the following:
\begin{itemize}
\item The choice of cells will be made in a ``uniform'' manner. This part is helped by $\calb'$ being a syntactic refinement. We will in fact set the ``subcell ID'' (the values of the polynomials appearing in $\calb'$ and not in $\calb$) to be the same for all cells of $\calb$.
\item All the subcells will feature a ``good'' decomposition, in terms of the norm of $f_3$.
\item Most subcells will ``well-represent'' their corresponding cells from $\calb$, in terms of the corresponding conditional expectation of $f$.
\end{itemize}

Now we state this formally.

\begin{corollary}[Subcell Selection]\label{cor:subatom}
Suppose $\zeta > 0$ and $d \geq 1$ is an integer less than $p$. Let $\eta, \delta: \N \to \R^+$ be arbitrary
non-increasing functions, and let $r: \N \to \N$ be an arbitrary
non-decreasing function. Then, there exist $C =
C_{\ref{cor:subatom}}(\delta, \eta, p,  r, \zeta)$ such that the
following holds.

Given $f: \F_p^n \to \zo$, there exist functions $f_1, f_2, f_3 :
\F_p^n \to \R$, a polynomial factor $\calb$ with cells denoted by
elements of $\F_p^{|\calb|}$, a syntactic refinement $\calb'$ of
$\calb$ with complexity at most $C$ and cells denoted by elements of
$\F_p^{|\calb|} \times \F_p^{|\calb'| - |\calb|}$, and an element $s
\in \F_p^{|\calb'| - |\calb|}$ such that the following is true:
\begin{itemize}
\item
$f = f_1 + f_2 + f_3$
\item
$f_1 = \E[f| \calb']$
\item
$\|f_2\|_{U^{d+1}} < \eta(|\calb'|)$
\item
$f_1$ and $f_1 + f_3$ have range $[0,1]$; $f_2$ and $f_3$ have
range $[-1,1]$.
\item
$\calb$ is of rank at least $r(|\calb|)$
\item
$\calb'$ is of rank at least $r(|\calb'|)$
\item
For every $c \in \F_p^{|\calb|}$, the subcell $c' = (c, s) \in
\F_p^{|\calb'|}$ has the property that
$\E_{\calb(x) = c'}[(f_3(x))^2] < (\delta(|\calb|))^2$.
\item
$\Pr_{c \in \F_p^{|\calb|}}[|\E[f|c] - \E[f|(c,s)]|>\zeta]<\zeta$, where we denote $\E[f|c]=\E[f(x)|\calb(x)=c]$ and $\E[f|(c,s)]=\E[f(x)|\calb'(x)=(c,s)]$.
\end{itemize}
\end{corollary}

\begin{proof}
Let $r'(m)=r_{\ref{thm:rankreg}}(p,p^{-m}/10)$, so that by Theorem~\ref{thm:rankreg}, a polynomial factor $\calb$ of degree $d$ and rank at least $r'(|\calb|)$ satisfies for any $c \in \F_p^{|\calb|}$
$$
0.9 \; p^{-|\calb|} \le \Pr_{x \in \F_p^n}[\calb(x)=c] \le 1.1 \; p^{-|\calb|}.
$$

Set $C_{\ref{cor:subatom}}(\delta, \eta, p,  r, \zeta) =
C_{\ref{thm:superdecomp}}(\Delta,\eta,p,r'',\zeta/4, 1)$, where
$\Delta(m) =  0.1 \cdot \delta(m)/p^{m}$ and $r''(m)=\max(r(m),r'(m))$. Apply Theorem \ref{thm:superdecomp} with
$\calb_0$ being the trivial partitioning consisting of one cell. This
yields a factor $\calb$ with rank at least $r''(|\calb|)$, and a syntactic refinement $\calb'$ of $\calb$ with rank
at least $r''(|\calb'|)$. Let $s$ be a uniformly chosen random element from
$\F_p^{|\calb'| - |\calb|}$.

Observe that for every cell $c \in \F_p^{|\calb|}$ of $\calb$, at most
a $0.1 p^{-|\calb|}$ fraction of the subcells $c' \in \{c\} \times
\F_p^{|\calb'|-|\calb|}$ of $\calb'$  have $\E_x[(f_3(x))^2|c'] > \delta(|\calb|)^2$. To show this, assume on the contrary
that even for one cell $c \in \F_p^{|\calb|}$ this event does not occur, and denote by $S$
the set of cells $c'\in\F_p^{|\calb'|}$ of $\calb'$ that lie in $c$ for which $\E_x[(f_3(x))^2|c'] > \delta(|\calb|)^2$. By our assumption
$|S|\geq(0.1p^{-|\calb|})p^{|\calb'|-|\calb|}$, and then
$\|f_3\|_2^2 = \E_{x  \in   \F_p^n}[(f_3(x))^2] > \delta(|\calb|)^2 \Pr_{x \in \F_p^n}[\calb(x)=c \wedge \calb'(x)\in S] \ge
0.09 \;\delta(|\calb|)^2/p^{2|\calb|} >
\Delta(|\calb|)^2$, a contradiction to the guarantee of Theorem
\ref{thm:superdecomp}.

Hence, for any fixed $c$, the probability that $s$ is such that $\E_x[(f_3(x))^2|(c,s)] > \delta(|\calb|)^2$
is at most $0.1p^{-|\calb|}$. By the union bound, with probability at
least $3/4$, for every $c \in \F_2^{|\calb|}$ the subcell
$c' = (c,s)$ has the property that $\E_x[(f_3(x))^2|c'] \leq \delta(|\calb|)^2$.

Also, because $\calb'$ $\zeta/4$-represents $\calb$, the expected
number of cells $c$ for which $|\E[f|c] - \E[f|(c,s)]| > \zeta$ is
less than $\zeta/4  \cdot p^{|\calb|}$. So, by the Markov inequality, with probability at least
$3/4$ $$\Pr_{c \in \F_p^{|\calb|}}[|\E[f|c] - \E[f|(c,s)]|>\zeta]<\zeta$$.

We conclude that an $s$ exists with both the desired properties.
\end{proof}

\subsection{Extending to Multiple Functions}\label{subsec:multi}

The theorems so far referred to only a single function $f:\F_p^n\to\{0,1\}$. However, we
actually require decomposition theorems which work for several
functions $f^{(1)},\ldots,f^{(R)}:\F_p^n\to\{0,1\}$ simultaneously with a single polynomial factor; alternatively, this could be thought of as decomposing a single ``vector'' function $f:\F_p^n\to\{0,1\}^R$.

It is quite straightforward to adapt all the previous proofs to this framework. The main adaptation to be done is the following version of the definition of a density index.

\begin{definition}
The {\em density index} of a factor $\calb$ with respect to a vector function $f=(f^{(1)},\ldots,f^{(R)}):\F_p^n\to\{0,1\}^R$ is the sum of the squared $l_2$ norms of the conditional expectation of the $f^{(i)}$ functions, that is $\indd(\calb)=\sum_{i=1}^R\E\left[(\E[f^{(i)}|\calb])^2\right]$.

Given a function $h:\N\to\N$ and a real parameter $\gamma$, A factor $\calb$ is {\em $(h,\gamma)$-robust} (semantically) if there exists no $\calb'$ which is a semantic refinement of $\calb$ for which $|\calb'|\leq h(|\calb|)$ and $\indd(\calb')\geq\indd(\calb)+\gamma$.
\end{definition}

From here we can follow nearly the exact same arguments. The main difference is that now all resulting bounds will depend on $R$, starting with the multiple functions analog analog of $T_{\ref{lem:rankrob}}$, as the index is now bounded by $R$ rather than $1$. Eventually we can reach the following version of the subcell selection theorem.

\begin{theorem}[Subcell Selection -- Multiple Functions]\label{thm:subatom2}
Suppose $\zeta > 0$ and $d \geq 1$ is an integer less than $p$. Let $\eta, \delta: \N \to \R^+$ be arbitrary
non-increasing functions, and let $r: \N \to \N$ be an arbitrary
non-decreasing function. Then, there exist $C =
C_{\ref{thm:subatom2}}(\delta, \eta, p,  r, \zeta, R)$ such that the
following holds.

Given $f^{(1)},\dots,f^{(R)}: \F_p^n \to \zo$, there exist functions $f^{(i)}_1, f^{(i)}_2, f^{(i)}_3 :
\F_p^n \to \R$ for all $i \in [R]$, a polynomial factor $\calb$ with cells denoted by
elements of $\F_p^{|\calb|}$, a syntactic refinement $\calb'$ of
$\calb$ with complexity at most $C$ and cells denoted by elements of
$\F_p^{|\calb|} \times \F_p^{|\calb'| - |\calb|}$, and an element $s
\in \F_p^{|\calb'| - |\calb|}$ such that the following is true:
\begin{itemize}
\item
$f^{(i)} = f_1^{(i)} + f_2^{(i)} + f_3^{(i)}$ for every $i \in [R]$.
\item
$f_1^{(i)} = \E[f^{(i)}| \calb']$ for every $i \in [R]$.
\item
$\|f_2^{(i)}\|_{U^{d+1}} < \eta(|\calb'|)$ for every $i \in [R]$.
\item
For every $i \in [R]$, $f_1^{(i)}$ and $f_1^{(i)} + f_3^{(i)}$ have range
$[0,1]$, and  $f_2^{(i)}$ and $f_3^{(i)}$ have range $[-1,1]$.
\item
$\calb$ is of rank at least $r(|\calb|)$
\item
$\calb'$ is of rank at least $r(|\calb'|)$
\item
for every $c \in \F_p^{|\calb|}$, the subcell $c' = (c, s) \in
\F_2^{|\calb'|}$ has the property that
$\E_x[(f_3^{(i)}(x))^2~|~\calb'(x) = (c,s)] < (\delta(|\calb|))^2$
for every $i \in [R]$.
\item
$\Pr_{c \in \F_p^{|\calb|}}[\exists_{i\in [R]}|\E[f^{(i)}|c] -
\E[f^{(i)}|(c,s)]|>\zeta]<\zeta$, where we denote $\E[f|c]=\E[f(x)|\calb(x)=c]$ and $\E[f|(c,s)]=\E[f(x)|\calb'(x)=(c,s)]$.
\end{itemize}
\end{theorem}

\section{Counting and Testability}\label{sec:count}

\subsection{Counting Patterns inside Cells}\label{subsec:counting}
Let $\calb$ be a polynomial factor generated by the polynomials $P_1,
\dots, P_C: \F_p^n \to \F_p$, and let $b_1,\dots,b_m \in \F_p^C$
denote the images of $m$ cells of $\calb$. We will want to estimate
probabilities of the following form:
\begin{equation}\label{eqn:expec}
\Pr_{x_1,\dots,x_\ell}[\calb(a_1(x_1,\dots,x_\ell)) = b_1 \wedge
\calb(a_2(x_1,\dots,x_\ell)) = b_2 \wedge \cdots \wedge
\calb(a_m(x_1,\dots,x_\ell)) = b_m]
\end{equation}
where $(a_1,\dots,a_m)$ is an affine constraint of size $m$ on $\ell$
variables.  In Lemma \ref{lem:cellsize}, we analyzed the expectation
when $\ell = m = 1$ and $a_1(x_1) = x_1$. In order to deal with
the more general form,  let us re-express (\ref{eqn:expec}) in the
following way:
\begin{align}
&\Pr_{x_1,\dots,x_\ell}[\calb(a_1(x_1,\dots,x_\ell)) = b_1 \wedge
\cdots \wedge \calb(a_m(x_1,\dots,x_\ell)) = b_m]\nonumber \\
&= \E_{x_1,\dots,x_\ell \in \F_p^n}\left[\prod_{i \in [C]}
  \prod_{j \in [m]} \frac{1}{p}\sum_{\lambda_{i,j} \in \F_p}
  \expo{\lambda_{i, j}\cdot (P_i(a_j(x_1,\dots,x_\ell))
    - b_{i,j})}\right]\nonumber \\
&= p^{-m C} \sum_{\lambda_{i,j} \in \F_p : \atop i \in [C], j
  \in [m]}\expo{-\sum_{i \in [C]} \sum_{j \in [m]}
  \lambda_{i,j} b_{i,j}} \E_{x_1,\dots,x_\ell}\left[\expo{
    \sum_{i \in [C]} \sum_{j \in [m]}
    \lambda_{i,j}P_i(a_j(x_1,\dots,x_\ell))}\right] \label{eqn:orth}
\end{align}

Hatami and Lovett in \cite{HL11,HL11b} studied expectations such as those in
(\ref{eqn:orth}) and proved the following dichotomy.

\begin{lemma}[Lemma 5.1 in \cite{HL11b}]\label{lem:hl1}
Suppose we are given $\eps \in (0,1)$, positive integer $d<p$ and an affine
constraint $(A,\sigma)$ where $A=(a_1, \dots, a_m)$ is of size $m$ and over
$\ell$ variables. Let $P_1, \dots, P_C: \F_p^n \to \F_p$
be a collection of polynomials of degree at most $d$  such that the
rank of the polynomial factor generated by $P_1,\dots,P_C$ is at least
$r_{\ref{thm:rankreg}}(d,\eps)$. Then, for every set of coefficients
$\Lambda = \{\lambda_{i,j}\in \F_p: i \in [C], j \in [m]\}$, if
$P_\Lambda: (\F_p^n)^{\ell} \to \F_p$ is the
polynomial defined by:
$$ P_\Lambda(X_1, \dots, X_{\ell}) = \sum_{i=1}^C \sum_{j=1}^m
\lambda_{i,j} P_i\left(a_j(X_1,\dots,X_\ell)\right)$$
then either $P_\Lambda$ is the zero polynomial, or else $\left|\E_{x_1, \dots,
  x_\ell \in \F_p^n} \expo{P_\Lambda(x_1,\dots, x_{\ell})}\right| < \eps$.
\end{lemma}

Thus, to bound (\ref{eqn:orth}), we need to count the number of sets
$\Lambda$ such that $P_\Lambda \equiv 0$, in the language of Lemma
\ref{lem:hl1}. To this end, let us make the following definition,
following the works of Gowers and Wolf~\cite{gowers-wolf-1,gowers-wolf-2}.

\begin{definition}[Dimension of linear forms]
For a positive integer $d$ and linear form $L(X_1,\dots,X_\ell) =
\alpha_1X_1 + \alpha_2 X_2 + \cdots + \alpha_\ell X_\ell$ where
$\alpha_1,\dots,\alpha_\ell \in \F_p$, let the {\em $d$th tensor
  power of $L$} denote: $$L^{\otimes d}~\eqdef~\left(\prod_{j=1}^d
  \alpha_{i_j}~:~i_1,\dots,i_d \in [\ell]\right) \in \F_p^{\ell^d}$$

Given positive integers $d_1, \dots, d_C$ and an affine constraint
$A = (a_1, \dots, a_m)$ of size $m$ on $\ell$ variables, define the {\em
  $(d_1,\dots,d_C)$-dimension of $A$} to be:
$$\sum_{i=1}^C \dim\left(\left\{a_1^{\otimes d_i}, \dots, a_m^{\otimes d_i}\right\}\right)$$
\end{definition}

To show the relevance of the above definition, we first need an algebraic ``all or nothing'' lemma from \cite{HL11b} that concerns linear and polynomials without explicitly referring to the dimension of the forms.

\begin{lemma}[Lemma 5.2 in \cite{HL11b}]\label{lem:hl2}
Suppose $\lambda_{i,j}\in \F_p$ for $i \in [C], j \in [m]$, and $d_1,
\dots, d_C \in [d]$, where $d < p$. Also, let $(A,\sigma)$ where $A = (a_1, \dots,
a_m)$ be an affine constraint, where every linear form $a_j$ is over
variables $X_1,\dots,X_\ell$. Then, one of the following
holds:
\begin{itemize}
\item
For every collection of linearly independent polynomials $P_1, \dots,
P_C$ of degree $d_1, \dots, d_C$ respectively:
$$\sum_{i=1}^C \sum_{j=1}^m
\lambda_{i,j} P_i\left(a_j(X_1,\dots,X_\ell)\right)
\equiv 0$$
\item
For every collection of linearly independent polynomials $P_1, \dots,
P_C$ of degree $d_1, \dots, d_C$ respectively:
$$\sum_{i=1}^C \sum_{j=1}^m
\lambda_{i,j} P_i\left(a_j(X_1,\dots,X_\ell)\right) \not \equiv
0$$
\end{itemize}
\end{lemma}

Now we can make the connection between the definition of the dimension of the linear forms, and their effect on a sequence of polynomials with given degrees.

\begin{lemma}\label{lem:count}
Let the notation here be same as in Lemma \ref{lem:hl1}. If
$d_1,\dots,d_C$ are the respective degrees of the polynomials
$P_1,\dots, P_C$ and if $s$ is the $(d_1,\dots,d_C)$-dimension of
$(a_1,\dots,a_m)$, then the number of sets $\Lambda$ for which
$P_\Lambda \equiv 0$ equals $p^{mC-s}$.
\end{lemma}
\begin{proof}
Notice that we want to show that the number of sets $\Lambda$ for
which $P_\Lambda \equiv 0$ is dependent just on the degrees of the
polynomials $P_1,\dots,P_C$ and not on any other specifics. For this
we use Lemma \ref{lem:hl2}, so that instead of having the polynomials
$P_1,\dots,P_C$, we can analyze a collection of much simpler linearly
independent polynomials of respective degrees $d_1,\dots,d_C$.

In particular, let us define $P_i'(x) = x_i^{d_i}$ for every $i \in [C]$
(we assume that $n> C$). Then, the polynomial $P_\Lambda'(X_1, \dots,
X_{\ell}) = \sum_{i=1}^C \sum_{j=1}^m \lambda_{i,j}
P_i'\left(a_j(X_1,\dots,X_\ell)\right)$ is identically zero exactly
when $\sum_{j=1}^m \lambda_{i,j}a_j^{\otimes d_i} = 0$ for every $i \in [C]$.

Standard linear algebra and the definition of $(d_1,\dots,d_C)$-dimension
then shows that the set of $\Lambda$'s for which $P_\Lambda' \equiv 0$
forms a linear subspace of codimension $s$.
\end{proof}

At this point, we can move to the main theorem of this
section. Let us first make the following definition, that in some ways
captures the essence of ``polynomial feasibility'' for a sequence of values.
\begin{definition}\label{def:cons}
Given an affine constraint $A = (a_1,\dots,a_m)$ and positive integers
$d_1, \dots, d_C$, we say that elements $b_1,\dots,b_m$, where
$b_j = (b_{1,j},\dots,b_{C,j}) \in \F_p^C$ for every $j\in [m]$, are {\em consistent with
  respect to $A$ and $d_1,\dots,d_C$} if the following is true:
\begin{itemize}
\item For every set $\Lambda = \{\lambda_{i,j} \in \F_p : i \in [C], j \in [m]\}$
for which  $\sum_{j \in [m]}\lambda_{i,j}
(a_j(X_1,\dots,X_\ell))^{\otimes d_i}$ equals $0$ for all $i \in [C]$,
it is the case that  $\sum_{j \in [m]} \lambda_{i,j} b_{i,j} = 0$ as
well for all $i \in [C]$.
\end{itemize}
\end{definition}

The following is easy to observe using basic linear algebra:

\begin{observation}
Being consistent is equivalent to satisfying the following condition:
For every set $\Lambda = \{\lambda_{i,j} \in \F_p : i \in [C], j \in [m]\}$
for which  $\sum_{j \in [m]}\lambda_{i,j}
(a_j(X_1,\dots,X_\ell))^{\otimes d_i}$ equals $0$ for all $i \in [C]$,
we have  $\sum_{i\in [C]}\sum_{j \in [m]} \lambda_{i,j} b_{i,j} = 0$.
\end{observation}

The following theorem shows that the expectation in (\ref{eqn:expec}) is nonzero,
and is in fact close to a calculated number, if and only if $b_1,\dots,b_m$ are consistent.
\begin{theorem}\label{thm:density}
Let $\eps \in (0,1)$, let $(A,\sigma)$ where $A = (a_1,\dots,a_m)$ be an affine constraint over $\ell$
variables, and let $\calb$ be a polynomial factor of degree $d$, complexity $C$ and
rank at least $r_{\ref{thm:rankreg}}(d,\eps)$ generated by the
polynomials $P_1, \dots, P_C: \F_p^n \to \F_p$. For every $i \in [C]$, let
$d_i$ be the degree of $P_i$. Let $s$ denote the
$(d_1,\dots,d_C)$-dimension of $A$ over $\F_p$. Finally, for every $j
\in [m]$, fix the image of a cell in $\calb$, indexed by $b_j  =
(b_{1,j}, \dots, b_{C,j}) \in \F_p^C$.

If $b_1,\dots,b_m$ are consistent with respect to $A$ and
$d_1,\dots, d_C$, then:
\begin{equation*}
\Pr_{x_1, \dots, x_\ell \in
  \F_p^n}\left[\calb\left(a_1(x_1,\dots,x_\ell)\right)=b_1 \wedge\cdots\wedge
  \calb\left(a_m(x_1,\dots,x_\ell)\right)=b_m \right]
= p^{-s} \pm \eps
\end{equation*}
If $b_1,\dots,b_m$ are not consistent with respect to $A$ and
$d_1,\dots,d_C$, then the above probability is $0$.
\end{theorem}
\begin{proof}
Assume first that the supposition is true.
Let us rewrite the probability in question as in (\ref{eqn:orth}):
\begin{align*}
p^{-m C} \sum_{\lambda_{i,j} \in \F_p : \atop i \in [C], j
  \in [m]}\expo{-\sum_{i \in [C]} \sum_{j \in [m]}
  \lambda_{i,j} b_{i,j}} \E_{x_1,\dots,x_\ell}\left[\expo{
    \sum_{i \in [C]} \sum_{j \in [m]} \lambda_{i,j}P_i(a_j(x_1,\dots,x_\ell))}\right]
\end{align*}
According to Lemma \ref{lem:hl1}, the expectation in the above
expression is at most
$\eps$ in absolute value if
$\sum_{i \in [C]} \sum_{j \in [m]}
\lambda_{i,j}P_i(a_j(X_1,\dots,X_\ell))$ is
not the zero polynomial. On the other hand, by the argument of Lemma
\ref{lem:count}, if $\sum_{i \in [C], j \in [m]}
\lambda_{i,j}P_i(a_j(X_1,\dots,X_\ell))
\equiv 0$, then $\sum_{i \in [C], j \in [m]} \lambda_{i,j}
a_j^{\otimes d_i}$ equals $0$. Hence, in this case,
by consistency, $\sum_{i \in [C]}\sum_{j \in [m]}
\lambda_{i,j} b_{i,j} = 0$, and so, such a choice of
$\{\lambda_{i,j}\}$ contributes $1$ to the outermost summation.
The number of such choices of $\{\lambda_{i,j}\}$ is $p^{mC-s}$ by
Lemma \ref{lem:count}. Thus:
\begin{align*}
\Pr_{x_1, \dots, x_\ell \in
  \F_p^n}\left[\forall i \in [C], j  \in [m]~P_i\left(a_j(x_1,\dots,x_\ell)\right) = b_{i,j}
  \right]
= p^{-mC}(p^{mC-s} \pm p^{mC}\eps) = p^{-s} \pm \eps
\end{align*}

The last part of the Theorem follows easily. Suppose the probability
in question is nonzero, and so there exist $x_1,
\dots, x_\ell$ so that  $\calb\left(a_j(x_1,\dots,x_\ell)\right) =
b_{i,j}$ for all $i \in [C]$ and $j \in
  [m]$. Then, for all possible values of $\lambda_{i,j}$ we have $\sum_{i \in [C], j  \in
    [m]}\lambda_{i,j}b_{i,j} = \sum_{i \in [C],j \in
    [m]}\lambda_{i,j}P_i\left(a_j(x_1,\dots,x_\ell)\right)$.   But, by
  the argument of   Lemma \ref{lem:count}, $\sum_{j \in
    [m]}\lambda_{i,j}P_i\left(a_j(X_1,\dots,X_\ell)\right) \equiv 0$
  if  $\sum_{ j \in [m]}\lambda_{i,j}
  (a_j^{\otimes d_i}) = 0$ for any $i \in [C]$, and so the supposition
  is true.
\end{proof}

By now we know the importance of the definition of consistency. One more building block that we need shows why, when selecting cells from a refining partition as in Theorem \ref{thm:superdecomp}, consistency will pass over from $\calb$ to $\calb'$.

\begin{lemma}\label{lem:present}
Suppose that $\calb'$ is a syntactic refinement of $\calb$, that $A=(a_1,\ldots,a_m)$
is a sequence of linear forms where $(A,\sigma)$ is some affine constraint, and that $c_1,\ldots,c_m\in\F_p^{|\calb|}$ are
consistent with $A$ and the degrees $d_1,\ldots,d_{|\calb|}$ of the polynomials defining $\calb$.
Given any fixed $s\in\F_p^{|\calb'|-|\calb|}$, the cells $c_1',\dots,c_m'\in\F_p^{|\calb'|}$
defined by the concatenations $c'_j=(c_j,s)$ for all $j\in [m]$ are consistent with respect to $A$
and the degrees $d_1,\dots,d_{|\calb'|}$ of the polynomials defining $|\calb'|$.
\end{lemma}
\begin{proof}
Coordinate-wise, for $j \in [m]$, $c_j' =
(c_{1,j}',c_{2,j}',\dots,c_{|\calb'|, j}') \in \F^{|\calb'|}$ satisfies
 $c_{i,j}' = c_{i,j}$ for all $1\leq i \leq |\calb|$ and $c_{i,j}' =
 s_{i-|\calb|+1}$ for all $|\calb| < i \leq |\calb'|$.

To show the consistency condition for $i> |\calb|$, recall
that each $a_j$ is of the form $X_1 + \sum_{r=2}^\ell c_{r} X_r$ for
$c_r \in \F_p$ (as it came from an affine constraint). So, whenever $\sum_{j \in
  [m]}\lambda_{i,j} (a_j)^{\otimes d_i} = 0$ for any $d_i >0$,
we have that $\sum_{j \in [m]}\lambda_{i,j} = 0$, simply by looking at the sum
along the  coordinate of $a_j^{\otimes d_i}$ corresponding to
$X_1^{\otimes   d_i}$ (i.e. the one labeled by the sequence $(1,\ldots,1)$;
the vector $a_j^{\otimes d_i}$ will always be $1$ at that coordinate).
Since for any $i > |\calb|$, $c'_{i,j}=s_{i-|\calb|+1}$ is
independent of $j$, it follows that for any $i > |\calb|$, if  $\sum_{j \in
  [m]}\lambda_{i,j} (a_j)^{\otimes d_i} = 0$, then  $\sum_{j \in
    [m]}\lambda_{i,j}c'_{i,j} = s_{i-|\calb|+1} \sum_{j \in
    [m]}\lambda_{i,j} = 0$ for all $i>|\calb|$. Since we already know
(by the consistency of the $c_i$ relative to $\calb$) that $\sum_{j \in [m]}\lambda_{i,j}c'_{i,j}=0$
for all $i\leq|\calb$|, we can conclude the proof.
\end{proof}

\subsection{Big Picture Arguments}\label{subsec:big}

We will prove the existence of many copies of a given linear constraint by analyzing the existence of a particular configuration of cells of a factor $\calb$, where in every cell we look at the entire set of values that $f$ can take at once. The following is a formal definition of the function giving the ``big picture''.

\begin{definition}
Given a function $f:\F_p^n\to [R]$ and a polynomial factor $\calb$, the {\em big picture function} of $f$ is the function $f_{\calb}:\F_p^{|\calb|}\to 2^{[R]}$, where $2^{[R]}$ denotes the power set of $R$, defined by $f_{\calb}(y)=\{f(x):\calb(x)=y\}$. In other words, $f_{\calb}(y)$ is the set of all values that $f$ takes within the corresponding cell of $\calb$.
\end{definition}

On the other hand, given any function $g:\F_n^C\to 2^{[R]}$, and a set of degrees $d_1,\ldots,d_C$ (of which we think as corresponding to the degrees of some future polynomial factor of size $C$), we will define what it means for such a function to ``induce'' a copy of a given constraint.

\begin{definition}[Partially induce]\label{def:partial}
Suppose we are given positive integers $d_1, \dots, d_C$, a
function
$g: \F_p^C \to 2^{[R]}$, and an induced affine constraint
$(A,\sigma)$ of size $m$ over $\ell$ variables. We say that {\em $g$ partially
  $(d_1,\dots,d_C)$-induces $(A,\sigma)$} if there exist $\{b_j =
(b_{1,j}, \dots, b_{C,j})\in \F_p^C :
j \in [m]\}$ making the following true.
\begin{itemize}
\item
$b_1,\dots,b_m$ are consistent with respect to $A$ and $d_1,\dots,d_C$.
\item
 $\sigma_j \in g(b_j)$ for every
$j \in [m]$.
\end{itemize}
\end{definition}

The big picture function defined above extracts a finitary description of
a function $f:\F_p^n\to [R]$ in relation to some $\calb$, which we will
later obtain through a decomposition theorem. Regardless of how we obtained
$\calb$, moving from an induced constraint of $f$ to a partially induced constraint
of the big picture function $f_{\calb}$ is always guaranteed.

\begin{observation}\label{obs:pinduce}
If $f:\F_p^n\to [R]$ induces a constraint $(A,\sigma)$, then for a factor $\calb$ with degree sequence $(d_1,\ldots,d_{|\calb|})$ (where all degrees are smaller than $p$), the function $f_{\calb}:\F_p^C\to 2^{[R]}$ partially $(d_1,\ldots,d_{|\calb|})$-induces $(A,\sigma)$.
\end{observation}

\begin{proof}
Let $m$ be the size of $A$ and $\ell$ be its number of variables.
Suppose that $F$ induces $(A,\sigma)$
at $x_1,\dots,x_{\ell}$, and let $c_1,\dots, c_{m} \in \F_p^{|\calb|}$ be the images
of the $m$ cells in $\calb$ defined by $c_1 =
\calb(a_{1}(x_1,\dots,x_{\ell})), c_2 =
\calb(a_{2}(x_1,\dots,x_{\ell})),$ $\dots,$ $c_{m} =
\calb(a_{m}(x_1,\dots,x_{\ell}))$ where $A = (a_{1},\dots,
a_{m})$. Then, because of the last condition in Theorem
\ref{thm:density}, it must be the case that $c_1,\dots,c_m$ are
consistent with respect to $A$ and $d_1,\dots,d_{|\calb|}$. This fulfills the first condition of Definition \ref{def:partial},
and the second condition is true by the definition of every $f_{\calb}(c_i)$ including all
values that $f$ takes in that cell.
\end{proof}

To handle a possibly infinite collection $\cala$ of affine constraints, we will employ a
compactness argument, analogous to one used in \cite{AS08a} to bound
the size of the constraint partially induced by the big picture function. Let us  make the
following definition:

\begin{definition}[The compactness function $\Psi_\cala$]\label{def:psi}
Suppose we are given a positive integer $C$ and a possibly infinite
collection of induced affine constraints $\mathcal{A} = \{(A^1,
\sigma^1), (A^2, \sigma^2), \dots\}$, where each
affine constraint $(A^i,\sigma^i)$ is of size $m_i$ and of complexity at
most $d<p$. For fixed $d_1, \dots, d_C < p$, denote by $\calg(d_1,\dots,d_C)$ to
be the set of functions $g: \F_p^C \to 2^{[R]}$ that
partially $(d_1,\dots,d_C)$-induce some $(A^i,\sigma^i)\in
\cala$. Now, we define the following function:
\begin{align*}
\Psi_\cala(C) = \max_{d_1, \dots, d_C < p} \max_{g \in
  \calg(d_1, \dots, d_C)} \min_{(A^i,\sigma^i) \text{ {\em partially}} \atop
  \text{{\em induced by }} g} m_i
\end{align*}
Whenever $\calg(d_1, \dots, d_C)$ is empty we set the corresponding maximum to $0$.
\end{definition}

Note that the above is indeed finite, as both the number of possible degree sequences
(bounded by $p^C$) and the size of $\calg(d_1,\dots,d_C)$ (bounded by $2^{|R|p^C}$) are finite.
The compactness function allows to bound an induced constraint in advance,
at least (for now) in the realm of big picture functions:

\begin{observation}\label{obs:induce}
Let $d_1,\dots,d_C<p$ be a degree sequence, for which a function $g: \F_p^C \to 2^{[R]}$
partially induces some constraint from $\cala$. Then $g$ will necessarily
partially  induce some $(A^i,\sigma^i) \in
\cala$ whose size is at most $\Psi_\cala(|\calb|)$.
\end{observation}
\begin{proof}
This is immediate, as a $g$ satisfying the above in particular belongs to $\calg(d_1,\dots,d_C)$.
\end{proof}

For our proofs, we will refer first not to $f$ itself, but to some small modification of $f$ that will make it a ``perfect'' representation of some cells from $f$ according to some factor, which will be selected as per Corollary \ref{cor:subatom}.

\begin{definition}[Function cleanup] \label{def:clean}
Suppose we have a factor $\calb'$ that is a syntactic refinement of $\calb$, and some $s\in\F_n^{|\calb'|-|\calb|}$.
The {\em $\zeta$-cleanup} $F$ of $f:\F_p^n\to [R]$ according to $\calb$, $\calb'$ and $s$ is constructed by executing the
following steps in order (where as usual $(c,s)$ denotes the concatenation of $c$ and $s$):
\begin{enumerate}
\item
For every $z \in \F_p^n$ that is not covered by the cases below, let $F(z) = f(z)$.
\item
For every cell $c$ of $\calb$ for which $|\Pr[f(x)=i~|~c]-\Pr[f(x)=i~ |~(c,s)]|>\zeta$
for any $i \in [R]$, do
the following. For every $z \in \calb^{-1}(c)$, set $F(z) = \arg\max_{j \in
  [R]}\Pr[f(x)=j~|~(c,s)]$, the most popular value inside the
subcell $(c,s)$ (breaking ties arbitrarily, but consistently within each cell $c$).
\item
For every cell $c$ of $\calb$, for every $i \in [R]$ such that
$\Pr[f(x)=i~|~(c,s)]<\zeta$, set $F(z) =  \arg\max_{j \in
  [R]}\Pr[f(x)=j~|~(c,s)]$ for every $z \in f^{-1}(i)\cap
\calb^{-1}(c)$ (breaking ties arbitrarily, but consistently within each cell $c$).
\end{enumerate}
\end{definition}

\begin{lemma}\label{lem:cleanup}
If $f$, $\calb$, $\calb'$ and $s$ are such that $\calb$ is of rank at least $r_{\ref{thm:rankreg}}(p,\beta/p^{|\calb|})$, and $\Pr_{c \in \F_p^{|\calb|}}[|\E[f|c] - \E[f|(c,s)]|>\zeta]<\zeta$, then the corresponding $\zeta$-cleanup $F$ is $(2R+1+\beta)\zeta$-close to $f$.
\end{lemma}
\begin{proof}
Observe that the second step changes the value of $F$ on at most
a $\zeta$ fraction of the cells, by the condition involving $s$
in the statement of the lemma. By Lemma \ref{lem:cellsize}, each cell
occupies at most a $(1+\beta)p^{-C}$ fraction of the entire
domain. So, the fraction of points whose values changed in the second
step is at most $\zeta p^C\cdot (1+\beta)p^{-C}=(1+\beta)\zeta$.

The third step does not apply to any cell of
$\calb$ affected by the second step. Therefore, in the third case,
for every $i\in [R]$, if $\Pr[f(x)=i~|~\calb'(x) = (c,s)]<\zeta$ then
$\Pr[f(x)=i~|~\calb(x) = c]<2\zeta$. Hence, the total fraction of the
domain modified in the third case is at most $2R\zeta$. The total distance of
$F$ from $f$ is therefor bounded by $(2R+1+\beta)\zeta$.
\end{proof}

\subsection{More about Algebra of Linear Forms}\label{subsec:more}

A linear form $a(X_1,\ldots,X_\ell)=\sum_{i=1}^{\ell}\alpha_i X_i$ can be identified with a linear function over $\F_p^\ell$, and thus a transformation in the spirit of ``a change of basis'' can be formulated.

\begin{definition}[Change of view]
We identify the form $a(X_1,\ldots,X_\ell)=\sum_{i=1}^{\ell}\alpha_i X_i$ with the linear function $a:\F_p^\ell\to\F_p$ given by $a(v)=\sum_{i=1}^{\ell}\alpha_i v_i$, where $v=(v_1,\ldots,v_\ell)\in\F_i^\ell$ (in essence this is obtained by letting the $X_i$ range over scalars from $\F_p$ rather than vectors from some space $\F_p^n$).

Given an invertible $\ell\times\ell$ matrix $M$ over $\F_p$, the corresponding {\em change of view} of $a$ is the linear form $a'(X_1,\ldots,X_\ell)=\sum_{i=1}^{\ell}\alpha'_i X_i$ obtained by the following process: Consider the linear function corresponding to $a$, perform on its domain $\F_p^\ell$ the change of variables corresponding to $M$, and then take the linear form corresponding its representation $a'$ in the new basis.
\end{definition}

The reason that we use the term ``change of view'' is to not confuse it with a change of basis of $\F_p^n$. The following observation is easy:

\begin{observation}
If $(A,\sigma)$ is an affine constraint, and $A'$ is obtained by performing the same change of view over all linear forms of $A$, then a function $f:\F_p^n\to\F$ satisfies $(A,\sigma)$ if and only is it satisfies $(A',\sigma)$.

Additionally, a change of view does not affect the complexity of the affine constraint.
\end{observation}

This yields the following lemma:

\begin{lemma}\label{lem:conc}
Any affine constraint $(A,\sigma)$ is equivalent to one whose number of variables is not more than the number constraints.
\end{lemma}

\begin{proof}
Assume that $A=(a_1,\ldots,a_m)$ take $\ell$ variables for $\ell>m$, and consider the linear functions from $\F_p^\ell$ to $\F$ corresponding to $a_1,\ldots,a_m$. By a linear dimension argument there are $\ell-m$ linearly independent vectors $u_1,\ldots,u_{\ell-m}\in\F_p^{\ell}$ for which $a_i(v_j)=0$ for all $i\in [m]$ and $j\in [\ell-m]$. Complete these vectors to a basis $u_1,\ldots,u_\ell$ of $\F_p^\ell$, making sure that $u_\ell$ equals the vector that is $1$ on its first coordinate and zero everywhere else (this vector is not in the span of $u_1,\ldots,u_{\ell-m}\in\F_p^{\ell}$, because by the definition of an affine constraint $a_1$ sends it to $1$).

Now perform on the members of $A$ the change of view corresponding to the change to this basis of $\F_p^\ell$. Denoting the resulting linear forms by $A'=(a'_1,\ldots,a'_m)$, we note now that no $a'_i$ has any mention of the variables $X_1,\ldots,X_{\ell-m}$, and so the constraint $(A',\sigma)$ in fact takes at most $m$ variables. $A'$ will also have the standard form of an affine constraint with $X_\ell$ taking the place of $X_1$.
\end{proof}

We need the above because the test would eventually query a number of places that is a function of $p$ and the maximum number of variables in a subset of the constraints of $\cala$, where this subset is only guaranteed a bound on the number of linear forms per constraint; we thus need $\cala$ to satisfy the following definition:

\begin{definition}[Concise collections] \label{def:concise}
The collection $\mathcal{A} = \{(A^1, \sigma^1), (A^2, \sigma^2),\ldots\}$ is called {\em concise} if for every $A_i$, the total number of its variables does not exceed the number of its linear forms.
\end{definition}

Lemma \ref{lem:conc} implies that every collection of linear constraints is equivalent to a concise one.

We would also need to know the (lack of) affect that a change of view has on the $d$-dimension, and hence the $(d_1,\ldots,d_C)$-dimension, of $A$.

\begin{lemma}\label{lem:samedim}
If $A=(a_1,\ldots,a_m)$ is a sequence of linear forms, and $A'=(a'_1,\ldots,a'_m)$ is a sequence of the resulting forms after a fixed change of view, then $A$ and $A'$ have the same $d$-dimension for any $d$.
\end{lemma}

\begin{proof}
We use the identification of linear forms with linear functions from $\F_p^\ell$ to $\F_p$, and by extension for a linear form $a$ we consider the vector $a^{\otimes d}$ as the multilinear function $a^{\otimes d}:(\F_p^\ell)^d\to\F_p$ that sends $(v^{(1)},\ldots,v^{(d)})$ to $\prod_{i=1}^d a(v^{(i)})$; the representation of this multilinear function in the standard basis indeed corresponds to the vector originally defined as $a^{\otimes d}$.

The operation that takes $a$ to $a^{\otimes d}$ is not linear in itself; however, a change of basis over $\F_p^\ell$ (corresponding to the change of view) can be extended to an invertible linear operation over the linear space of all multilinear functions of $d$ vectors (not all of which come from linear forms). Namely, if $M$ is the basis change matrix, then the change of view for $a$ sends it to the function defined by $a'(v)=a(Mv)$, and $(a')^{\otimes d}$ in fact corresponds to $\prod_{i=1}^d a(Mv^{(i)})$. Now by basic linear algebra, the operation that sends any multilinear form $b:(\F_p^\ell)^d\to\F_p$ to the form $b'$ defined by $b'(v^{(1)},\ldots,v^{(d)})=a(Mv^{(1)},\ldots,Mv^{(d)})$ is linear and invertible; thus the $d$-dimension, and in fact the exact corresponding linear dependencies, do not change when moving from $A=(a_1,\ldots,a_m)$ to $A'=(a'_1,\ldots,a'_m)$.
\end{proof}

We end this section with a lemma about a ``juxtaposition'' of two sets of identical forms while sharing one variable.

\begin{lemma}\label{lem:doubledim}
Suppose that $(a'_1,\ldots,a'_m)$ are linear forms over $(X_1,\ldots,X_\ell)$ of $d$-dimension $q$, where for some $k$ the form $a'_k$ sends $(X_1,\ldots,X_\ell)$ to $X_1$.
The $d$-dimension of the following $2m$ linear forms over $(Z,X_2,\ldots,X_\ell,Y_2,\ldots,Y_\ell)$:
$$(a'_1(Z,X_2,\dots,X_\ell), \dots,a'_m(Z,X_2, \dots,X_\ell),
a'_1(Z, Y_2, \dots, Y_\ell),\dots,a'_m(Z,Y_2,\dots,Y_\ell))$$
is exactly $2q-1$.
\end{lemma}
\begin{proof}
We note that $a'_k(Z,X_2,\ldots,X_{\ell})=a'_k(Z,Y_2,\ldots,Y_{\ell})=Z$, and that all other
linear forms are distinct. Abusing notation somewhat, we let $Z$ denote also the linear form that returns the value of $Z$ from the variables $(Z,X_2,\ldots,X_\ell,Y_2,\ldots,Y_\ell)$; note that in particular $Z^{\otimes d}$ corresponds to the vector from $\F_p^{(2\ell-1)^d}$ that is $1$ on its coordinate corresponding to $(1,\ldots,1)$, and zero everywhere else.

Let $S \subseteq \{1,\ldots,m\}\setminus\{k\}$ be a set of size $q-1$
such that $\left\{\left(a'_j(Z, X_2, \dots, X_m)\right)^{\otimes d} : j \in S \cup \{k\} \right\}$ is a basis of size $q$ for the linear space $\textrm{span}\left\{\left(a'_j(Z, X_2, \dots, X_m)\right)^{\otimes d} : j \in [m] \right\}$. Clearly, $\left\{\left(a'_j(Z, Y_2, \dots, Y_m)\right)^{\otimes d} : j \in S \cup \{k\} \right\}$ is a basis for $\textrm{span}\left\{\left(a'_j(Z, Y_2, \dots,Y_m)\right)^{\otimes d} : j \in [m] \right\}$. Thus, the $d$-rank of the $2m$ linear forms is at most $2q -1$. To conclude, we will show that the $d$-rank is at least $2q-1$. To this end, we analyze the intersection
 $$\textrm{span}\left\{\left(a'_j(Z, X_2, \dots,   X_m)\right)^{\otimes d} : j \in S \cup \{k\} \right\}\cap \textrm{span}\left\{\left(a'_j(Z, Y_2, \dots,
    Y_m)\right)^{\otimes d} : j \in S \cup \{k\} \right\}.$$
It is clearly contained in $\textrm{span}\left\{Z^{\otimes d}\right\}$, since no other coordinate
can be non-zero in both sets (the left set can have only non-zero coordinates corresponding to sequences of length $d$
over $\{1,\ldots,\ell\}$, and the right set can have only non-zero coordinates corresponding to sequences of length $d$
over $\{1,\ell+1,\ldots,2\ell-1\}$). On the other hand, the intersection contains (and hence is equal to) $\textrm{span}\left\{Z^{\otimes d}\right\}$, because this vector appears on both sides (as $a'_k$).
This shows by a linear dimension argument that
the $d$-dimension of the $2m$ linear forms is exactly $2q -1$ as claimed.
\end{proof}

\subsection{The Proof of Testability}\label{subsec:proof}

We finally have all the building blocks in place to prove Theorem \ref{thm:main2}, which implies Theorem \ref{thm:main}.

\begin{proofof}{Theorem \ref{thm:main2}}
We begin with some
preliminaries. Let $d$ be the maximum complexity of an
affine constraint $A^i$ appearing in $\cala$. By hypothesis, $d <
p$. For $i \in [R]$, define $f^{(i)}:\F_p^n \to \zo$ so that
$f^{(i)}(x)$ equals $1$ when $f(x) = i$ and equals $0$
otherwise. Additionally, set the following parameters, where
$\Psi_\cala: \Z^+ \to \Z^+$ is the compactness function of $\cala$.
\begin{align*}
\alpha(C) &= p^{-2\Psi_A(C)C}\\
\rho(C) &= r_{\ref{thm:rankreg}}(d,\alpha(C))\\
\Delta(C) &= \frac{1}{16}\left(\frac{\eps}{8R}\right)^{\Psi_\cala(C)}\\
\eta(C) &= \frac{1}{8(3p)^{C\Psi_\cala(C)}}
\left(\frac{\eps}{8R}\right)^{\Psi_\cala(C)}\\
\zeta &= \frac{\eps}{8R}
\end{align*}
$\ell_{\cala}$ and $\delta_{\cala}$ will be defined, based on the above functions,
in (\ref{eqn:ell}) and (\ref{eqn:delta}) below.

Next, apply Theorem \ref{thm:subatom2} to the functions
$f^{(1)},f^{(2)}, \dots, f^{(R)}$ in order to get polynomial factors
$\calb' \preceq_{syn} \calb$ of degree $d$ and size at most
$C_{\ref{thm:subatom2}}(\Delta, \eta, p, \rho, \zeta, R)$, an element $s \in
\F_p^{|\calb'|-|\calb|}$, and functions
$f_1^{(i)}, f_2^{(i)}, f_3^{(i)}: \F_p^n \to \R$ for every $i \in
[R]$. The sequence of polynomials generating $\calb'$ will be denoted
by $P_1,\dots,P_{|\calb'|}$. Since $\calb'$ is a syntactic refinement,
$\calb$ is generated by the polynomials $P_1,\dots,P_{|\calb|}$.

Let $F$ be the $\zeta$-cleanup of $f$ with respect to $\calb$, $\calb'$ and $s$.
By Lemma \ref{lem:cleanup}, and what we know of these partitions and $s$,
$F$ is $\epsilon/2$-close to $f$, and hence by our assumption on the farness
of $f$, the function $F$ will still include an induced constraint from $\cala$.

By Observation \ref{obs:pinduce}, the big picture function $F_{\calb}$ of $F$
will $(d_1,\dots,d_{|\calb|})$-partially induce some constraint from $\cala$,
and hence by Observation \ref{obs:induce} it will partially induce some
$(A^i,\sigma^i)$ for which $m_i \leq \Psi_\cala(|\calb|)$. This will be the
constraint of which we will find many copies in the original $f$.
Let $m \eqdef m_i$, let $\ell~\eqdef~\ell_i$,
and let $\sigma_1, \dots, \sigma_m$ denote $\sigma^i_1,\dots,\sigma^i_m$
respectively. Since a concise $\cala$ means that $\ell_i\leq m_i$, we can now
define
\begin{equation}\label{eqn:ell}
\ell_{\cala}(\eps)=\Psi_\cala(C_{\ref{thm:subatom2}}(\Delta, \eta, p, \rho, \zeta, R)).
\end{equation}

Denote the linear forms in $A^i$ by $a_1, \dots, a_{m}$ and denote $\sigma^i=(\sigma_1,\ldots,\sigma_m)$.
Let $c_1 = (c_{1,1},\dots, c_{|\calb|,1}),\dots,c_{m} =
(c_{1,m},\dots, c_{|\calb|,m})\in \F_p^{|\calb|}$ index the cells of
$\calb$ where $(A^i,\sigma^i)$ is partially induced by $F_{\calb}$, the big picture function
of the cleanup function $F$, i.e., $c_1,\ldots,c_m$ are consistent,
and $\sigma_i\in F_{\calb}(c_i)$ for every $j \in [m]$. Also, let
$c_1',\dots,c_{m}' \in \F_p^{|\calb'|}$ index the associated subcells
of $\calb'$, obtained by letting $c_j' = (c_j,s)$ for every $j \in
[m]$.

Our goal will now be to lower bound:
\begin{align}
& \Pr_{x_1,\dots,x_{\ell}\in \F_p^n}\left[f(a_1(x_1,\dots,x_\ell))=\sigma_1
  \wedge\cdots\wedge f(a_m(x_1,\dots,x_\ell))=\sigma_m\right] \nonumber\\
& =\E_{x_1,\dots,x_{\ell}\in \F_p^n}\left[f^{(\sigma_1)}(a_1(x_1,\dots,x_\ell))\cdots
  f^{(\sigma_m)}(a_m(x_1,\dots, x_\ell))\right]\label{eqn:obj}
\end{align}
The theorem obviously follows if  the above
expectation is more than the respective $\delta_\cala(\eps)$. We rewrite the
expectation as:
\begin{equation}\label{eqn:obj2}
\E_{x_1,\dots,x_{\ell}\in \F_p^n}\left[(f_1^{(\sigma_1)}+
  f_2^{(\sigma_1)} + f_3^{(\sigma_1)})(a_1(x_1,\dots,x_\ell))\cdots
  (f_1^{(\sigma_m)}+  f_2^{(\sigma_m)} + f_3^{(\sigma_m)})(a_m(x_1,\dots, x_\ell))\right]
\end{equation}

We can expand the expression inside the expectation as a sum of $3^m$
terms. The expectation of any term which is a multiple of $f_2^{(\sigma_j)}$
for any $j \in [m]$ has an absolute value upper bound of
$\|f_2^{(\sigma_j)}\|_{U^{d+1}} \leq \eta(|\calb'|)$, because of Lemma
\ref{lem:cnt} and the fact that the complexity of $A^i$ is
bounded by $d$. Hence, the expression (\ref{eqn:obj2}) is at least:
\begin{equation}\label{eqn:obj3}
\E_{x_1,\dots,x_{\ell}}\left[(f_1^{(\sigma_1)}+
  f_3^{(\sigma_1)})(a_1(x_1,\dots,x_\ell))\cdots
  (f_1^{(\sigma_m)}+  f_3^{(\sigma_m)})(a_m(x_1,\dots, x_\ell))\right]
- 3^{m} \eta(|\calb'|)
\end{equation}

Before we continue, to ease notation, for the rest of the proof we will now define an indicator function.
$\cali_{(a_1,\ldots,a_m)}^{(c'_1,\ldots,c'_m)}(x_1,\ldots,x_{\ell})$ will be set to $1$ if $\calb'(a_j(x_1,\dots,x_\ell)) = c'_j$ for every $j\in [m]$, and it will be set to $0$ otherwise.

Now, because of the non-negativity of $f_1^{(\sigma_j)}+f_3^{(\sigma_j)}$
for every $j \in [m]$, the expectation in (\ref{eqn:obj3}) is at
least:
\begin{equation*}
\E_{x_1,\dots,x_{\ell}}\left[
\left(f_1^{(\sigma_1)}+  f_3^{(\sigma_1)}\right)(a_1(x_1,\dots,x_\ell))\cdots
  \left(f_1^{(\sigma_m)}+  f_3^{(\sigma_m)}\right)(a_m(x_1,\dots, x_\ell))\cdot 
\cali_{(a_1,\ldots,a_m)}^{(c'_1,\ldots,c'_m)}(x_1,\ldots,x_{\ell})
\right]
\end{equation*}
In other words, what we are doing now is
counting only patterns that arise from the selected subcells $c_1',
\dots, c_m'$. We next expand the product inside the expectation into $2^m$
terms. The main contribution will come from:
\begin{equation}\label{eqn:obj5}
\E_{x_1,\dots,x_{\ell}}\left[
f_1^{(\sigma_1)}(a_1(x_1,\dots,x_\ell))\cdots
 f_1^{(\sigma_m)}(a_m(x_1,\dots, x_\ell))\cdot
\cali_{(a_1,\ldots,a_m)}^{(c'_1,\ldots,c'_m)}(x_1,\ldots,x_{\ell})
\right]
\end{equation}

But first, let us show that the contribution from each of the other
$2^m - 1$ terms is small. Consider a term that contains
$f_3^{(\sigma_k)}$ for some $k \in [m]$. Letting $g$ denote an arbitrary
function with $\|g\|_\infty \leq 1$, such a term is of the form:
\begin{equation}\label{eqn:obj6}
\E_{x_1,\dots,x_\ell}\left[f_3^{(\sigma_k)}(a_k(x_1,\dots,x_\ell))
  g(x_1,\dots,x_\ell) \cdot
\cali_{(a_1,\ldots,a_m)}^{(c'_1,\ldots,c'_m)}(x_1,\ldots,x_{\ell})
\right]
\end{equation}
By our definition of affine constraints, $a_k(x_1,\dots,x_\ell)$
is of the form $x_1 + \sum_{i \in [\ell]} \alpha_i x_i$ for some
$\alpha_i \in \F_p$. We now change the summation variables
of the expectation by replacing $x_1$ with $z=x_1 + \sum_{i \in [\ell]} \alpha_i x_i$,
affecting a change of view for $a_1,\ldots,a_m$. Letting $a_1',\dots,a_m'$ denote
the linear forms as they appear after the change, we first note that $a'_k(Z,X_2,\ldots,X_{\ell})$
will equal $Z$. We can now bound the
square of (\ref{eqn:obj6}) using Cauchy-Schwarz as:
{\allowdisplaybreaks
\begin{align}
&\left(\E_{x_1,\dots,x_\ell}\left[f_3^{(\sigma_k)}(a_k(x_1,\dots,x_\ell))
  g(x_1,\dots,x_\ell) \cdot
\cali_{(a_1,\ldots,a_m)}^{(c'_1,\ldots,c'_m)}(x_1,\ldots,x_{\ell})
\right]\right)^2\nonumber\\
&\leq \left(\E_{z,x_2,\dots,x_\ell}\left[\left|f_3^{(\sigma_k)}(z)\right|
 \cdot
\cali_{(a'_1,\ldots,a'_m)}^{(c'_1,\ldots,c'_m)}(z,x_2,\ldots,x_{\ell})
\right]\right)^2\nonumber\\
&\leq \E_{z}\left[|f_3^{(\sigma_k)}(z)|^2 \cdot
  \cali_{(\mathrm{id})}^{(c'_k)}(z)\right] \cdot
\E_{z}\left[\left(\E_{x_2,\dots,x_\ell}\left[\cali_{(a'_1,\ldots,a'_m)}^{(c'_1,\ldots,c'_m)}(z,x_2,\ldots,x_{\ell})\right]\right)^2\right]\nonumber\\
&\leq \Delta^2(|\calb|)\cdot \Pr_{z}[\calb'(z) = c_k'] \cdot
\E_{z}\left[\left(\E_{x_2,\dots,x_\ell}\left[\cali_{(a'_1,\ldots,a'_m)}^{(c'_1,\ldots,c'_m)}(z,x_2,\ldots,x_{\ell})\right]\right)^2\right]\nonumber\\
&\leq \Delta^2(|\calb|) \cdot (p^{-|\calb'|}+\alpha(|\calb'|)) \cdot
\E_{z}\left[\left(\E_{x_2,\dots,x_\ell}\prod_{i \in [|\calb'|] \atop j \in
    [m]} \frac{1}{p}\sum_{\lambda_{i,j}\in
    \F_p}\expo{\lambda_{i,j}\cdot(P_i(a_j'(z,x_2,\dots,x_\ell))-c'_{i,j})}\right)^2\right]\nonumber\\
&\leq \frac{2\Delta^2(|\calb|)}{p^{2|\calb'|m + |\calb'|}}\E_{z}\left[\left(\sum_{\lambda_{i,j}\in \F_p:\atop i \in
    [|\calb'|], j \in [m]} \expo{-\sum_{i \in [|\calb'|]\atop j \in
      [m]} \lambda_{i,j}c'_{i,j}} \E_{x_2,\dots,x_\ell}\expo{\sum_{i \in [|\calb'|]\atop j \in
      [m]} \lambda_{i,j}P_i(a'_j(z,x_2,\dots,x_\ell))}\right)^2\right]\nonumber\\
&\leq \frac{2\Delta^2(|\calb|)}{p^{2|\calb'|m +
    |\calb'|}}\sum_{\lambda_{i,j}, \tau_{i,j} \in \F_p:\atop i \in
    [|\calb'|], j \in [m]} \left( \expo{-\sum_{i \in [|\calb'|]\atop j \in
      [m]} \lambda_{i,j}c'_{i,j}} \expo{\sum_{i \in [|\calb'|]\atop j \in
      [m]} \tau_{i,j}c'_{i,j}} \cdot \right. \nonumber \\
& \qquad\qquad\qquad \left. \E_{z,x_2,\dots, x_\ell \atop y_2, \dots, y_\ell}\left[\expo{\sum_{i \in [|\calb'|]\atop j \in
      [m]}
    \lambda_{i,j}P_i(a'_j(z,x_2,\dots,x_\ell))}\expo{-\sum_{i \in [|\calb'|]\atop j \in
      [m]}
    \tau_{i,j}P_i(a'_j(z,y_2,\dots,y_\ell))}\right]\right)\nonumber\\
&\leq \frac{2\Delta^2(|\calb|)}{p^{2|\calb'|m+|\calb'|}}\!\!\!\!\!\!\!\!\sum_{\lambda_{i,j},
  \tau_{i,j} \in \F_p:\atop i \in [|\calb'|], j \in [m]} \left|
  \E_{z,x_2,\dots, x_\ell \atop y_2, \dots, y_\ell}\left[\expo{\sum_{i \in [|\calb'|]\atop j \in
      [m]} \lambda_{i,j}P_i(a'_j(z,x_2,\dots,x_\ell))-\!\!\!\!\sum_{i \in [|\calb'|]\atop j \in
      [m]} \tau_{i,j}P_i(a'_j(z,y_2,\dots,y_\ell))}\right]\right|\label{eqn:obj7}
\end{align}}
Now, by Lemma \ref{lem:samedim}, the $(d_1,\dots,d_{|\calb'|})$-dimension of
$\{a_1,\dots,a_m\}$ equals the $(d_1,\dots,d_{|\calb'|})$-dimension of
$\{a_1',\dots,a_m'\}$.

Let $q$ denote the $(d_1,\dots,d_{|\calb'|})$-dimension of
$\{a_1,\dots,a_m\}$. By Lemma \ref{lem:doubledim}, summing over all of $(d_1,\dots,d_{|\calb'|})$,
we know that the $(d_1,\dots,d_{|\calb'|})$-dimension of
$$\left(a'_1(Z,X_2,\dots,X_\ell), \dots,a'_m(Z,X_2, \dots,X_\ell),
a'_1(Z, Y_2, \dots, Y_\ell),\dots,a'_m(Z,Y_2,\dots,Y_\ell)\right)$$
is exactly $q-|\calb'|$.

Now, just as in the proof of Theorem \ref{thm:density}, the
above information is enough to upper-bound (\ref{eqn:obj7}). 
The above $(d_1,\dots,d_{|\calb'|})$-dimension bound and Lemma \ref{lem:count} allow us to count
the number of $\lambda_{i,j}$ and $\tau_{i,j}$ such that the quantity
inside the expectation in (\ref{eqn:obj7}) is identically $1$, and Lemma \ref{lem:hl1}
along with the high-rank condition on the polynomials $P_i$  bounds
the expectation otherwise. It follows that 
(\ref{eqn:obj7}), and therefore the square of (\ref{eqn:obj6}), is at most: 
\begin{equation}\label{eqn:obj8}
\frac{2\Delta^2(|\calb|)}{p^{2m|\calb'|+|\calb'|}} \left(p^{2m|\calb'|-(2q-|\calb'|)} +
  p^{2m|\calb'|}\alpha(|\calb'|) \right) \leq 2\Delta^2(|\calb|)\cdot(p^{-2q} + \alpha(|\calb'|))
\end{equation}

Finally, we lower-bound the contribution from the main
term (\ref{eqn:obj5}). To begin with, we need to convince ourselves
that $f$ induces many copies of $(A^i,\sigma^i)$ among the
subcells $c_1',\dots,c_m'$. Recall
that $c_1,\dots,c_m$ are consistent with $d_1,\ldots,d_{|\calb|}$ and $A^i$,
and that $\sigma_i\in F_{\calb}(c_i)$ for every $i\in [m]$. By Lemma \ref{lem:present}
$c'_1,\dots,c'_m$ are consistent with $d_1,\ldots,d_{|\calb'|}$ and $A^i$ as well.

We can now lower-bound (\ref{eqn:obj5}) as follows:
\begin{align}
&\E_{x_1,\dots,x_{\ell}}\left[
f_1^{(\sigma_1)}(a_1(x_1,\dots,x_\ell))\cdots
 f_1^{(\sigma_m)}(a_m(x_1,\dots, x_\ell))\cdot
\cali_{(a_1,\ldots,a_m)}^{(c'_1,\ldots,c'_m)}(x_1,\ldots,x_{\ell})
\right] \nonumber\\
&= \Pr[\calb'(a_1(x_1,\dots,x_\ell)) = c'_1\wedge\cdots\wedge\calb'(a_m(x_1,\dots,x_\ell)) = c'_m] \cdot \nonumber\\
&\qquad\qquad \E_{x_1,\dots,x_{\ell}}\left[
f_1^{(\sigma_1)}(a_1(x_1,\dots,x_\ell))\cdots
 f_1^{(\sigma_m)}(a_m(x_1,\dots, x_\ell))|
\forall j\in [m] ~ \calb'(a_j(x_1,\dots,x_\ell)) = c'_j
\right]\nonumber\\
&\geq (p^{-q}-\alpha(|\calb'|))\cdot \left(\frac{\eps}{8R}\right)^m\label{eqn:obj9}
\end{align}
Let us justify the last line. The first term is due to Lemma
\ref{lem:present} and the lower bound on the probability from Theorem
\ref{thm:density}. The second term in (\ref{eqn:obj9}) is
because each $f_1^{(\sigma_j)}$ is constant on the cells of $\calb'$,
and because by construction, the big picture function $F_{\calb}$ of the cleanup function $F$, on which
$(A^i,\sigma^i)$ was partially induced, supports a value inside a cell $c$ of
$\calb$ only if the original function $f$ acquires the value on at
least an $\eps/(8R)$ fraction of the subcell $(c,s)$.

Combining the bounds from (\ref{eqn:obj3}), (\ref{eqn:obj8}) and
(\ref{eqn:obj9}), and using our parameter settings,  we get that
  (\ref{eqn:obj}) is at least:
\begin{align*}
&(p^{-q}-\alpha(|\calb'|))\cdot \left(\frac{\eps}{8R}\right)^m -
\sqrt{2\Delta^2(|\calb|)\cdot(p^{-2q} + \alpha(|\calb'|))} - 3^m\cdot
\eta(|\calb'|)\\
&>
\frac{p^{-q}}{2}\cdot\left(\frac{\eps}{8R}\right)^{\Psi_\cala(|\calb|)} -
2\Delta(|\calb|) \cdot p^{-q} - 3^{\Psi_\cala(|\calb|)}\cdot
\eta(|\calb'|)\\
&>\frac{p^{-\Psi_A(|\calb|)|\calb'|}}{4}\cdot\left(\frac{\eps}{8R}\right)^{\Psi_\cala(|\calb|)}
\end{align*}
where both $|\calb|$ and $|\calb'|$ are upper-bounded by
$C_{\ref{thm:subatom2}}(\Delta, \eta, p, \rho, \zeta, R)$ .
We can now define
\begin{equation}\label{eqn:delta}
\delta_{\cala}(\eps)=\frac14p^{-\Psi_A(C_{\ref{thm:subatom2}}(\Delta, \eta, p, \rho, \zeta, R))C_{\ref{thm:subatom2}}(\Delta, \eta, p, \rho, \zeta, R)}\cdot\left(\frac{\eps}{8R}\right)^{\Psi_\cala(C_{\ref{thm:subatom2}}(\Delta, \eta, p, \rho, \zeta, R))}
\end{equation}
to conclude the proof.
\end{proofof}

\bibliographystyle{alpha}
\bibliography{testing}

\end{document}